\documentclass[11pt]{article}
\usepackage[margin=1in]{geometry}
\usepackage{amsthm}
\newtheorem{theorem}{Theorem}
\newtheorem{lemma}{Lemma}
\newtheorem{corollary}{Corollary}
\theoremstyle{definition}
\newtheorem{definition}{Definition}

\usepackage{microtype}
\usepackage[T1]{fontenc}
\usepackage{amsmath}
\usepackage{amssymb}
\usepackage{comment}
\newcommand{\eps}{\epsilon}
\usepackage[hidelinks]{hyperref}
\usepackage{cleveref}
\usepackage{mdframed}
\newcommand{\chimax}{\chi_{\max}}
\newcommand{\Mint}{M_{\mathrm{int}}}
\usepackage{algorithm}
\usepackage{algpseudocode}
\newcommand{\OPT}{\mathit{OPT}}

\newcommand{\E}{\mathbb E}

\usepackage{xcolor}  

\usepackage[colorinlistoftodos,prependcaption,textsize=tiny]{todonotes}

\usepackage{newtxtext}       %
\usepackage[varvw]{newtxmath}       
\usepackage{xcolor}

\usepackage{graphicx}
\begin{document}
\title{Matroid Contention Resolution with Concentration}
\author{Marek Adamczyk\thanks{Institute of Computer Science, University of Wroc{\l}aw.}
\and Stephen Arndt\thanks{Carnegie Mellon University, Tepper School of Business. Supported in part by NSF grant CCF-2550399.}
\and Benjamin Moseley\thanks{Carnegie Mellon University, Tepper School of Business. Supported in part by NSF grant CCF-2550399.}
\and Kirk Pruhs\thanks{Department of Computer Science, University of Pittsburgh. Supported in part by NSF grants CCF-2550399 and CCF-2209654.}
\and Michael Zlatin\thanks{Department of Computer Science, Pomona College.}}
\date{}
\hypersetup{pdftitle={Matroid Contention Resolution with Concentration},
  pdfauthor={Marek Adamczyk, Stephen Arndt, Benjamin Moseley, Kirk Pruhs, Michael Zlatin}}
\maketitle

\begin{abstract}
Contention resolution schemes (CRS) are {widely used} for rounding
fractional solutions subject to combinatorial constraints. Their guarantees are typically in expectation.
Such guarantees alone do not ensure that covering constraints hold,
limiting the applicability to problems with covering requirements.

We prove lower-tail concentration bounds for the random-order CRS
of Adamczyk and W{\l}odarczyk (AW) under the intersection of $k$ matroid
constraints. For every linear function with coefficients in $[0,1]$,
the output attains nearly the fractional value divided by $k+1$
with failure probability exponentially small in this benchmark,
independently of the ground-set size.
We extend this bound to normalized monotone submodular
functions: the fractional value is given by the multilinear extension,
and the concentration bound accounts for the largest singleton value.
These guarantees apply to both objectives and covering requirements.
{
Our analysis bounds the fractional value lost in each update and applies
a general gain--loss concentration theorem whose bounds are independent
of the number of updates.
\par}

We apply these bounds to obtain new approximation guarantees for
matroid intersection coloring and for monotone submodular maximization
under matroid, packing, and covering constraints.
\end{abstract}

\section{Introduction}
Raghavan and Thompson~\cite{RaghavanThompson87} showed that independent
randomized rounding of fractional packing or covering LP solutions,
combined with Chernoff--Hoeffding bounds, yields near-optimal integral
solutions with high probability for many problems. Subsequent work on
\emph{dependent} rounding introduced correlations to enforce constraints
such as matroid independence, matching feasibility, and knapsack capacity
that independent rounding cannot guarantee. Examples include pipage
rounding~\cite{AgeevS04}, swap rounding~\cite{swap_rounding_one_matroid},
dependent rounding ~\cite{GandhiKPS06}, and the
schemes implicit in iterated rounding~\cite{Lau_Ravi_Singh_2011}.

A particularly significant development in this line of work was the shift from designing 
\emph{problem-specific} rounding algorithms to designing \emph{generic} rounding 
frameworks for families of problems. 
The \emph{contention resolution scheme} (CRS) framework introduced by
Chekuri, 
Vondr\'{a}k and Zenklusen~\cite{ChekuriVondrakZenklusen2014}
{applies to a down-closed family of feasible sets
$\mathcal{I}\subseteq 2^U$ on a ground set $U$, using a polytope
relaxation $P\subseteq[0,1]^U$ with
$\operatorname{conv}\{\mathbf1_S:S\in\mathcal{I}\}\subseteq P$.}

\begin{definition}
The input to a {\bf Contention Resolution Scheme (CRS)} $\pi$ is a fractional point $\mathbf{x} \in P$ and a (not necessarily feasible) subset $ A \subseteq U$ of elements, and the output is 
 a feasible subset $S$ of  both $A$  and the support of $x$, that is
$S \subseteq A \cap 
\mathrm{support}(\mathbf{x})$ and $S \in \mathcal{I}$. 
\end{definition}

CRS underlie 
state-of-the-art approximation results for many ``packing'' problems involving submodular maximization under matroid and knapsack 
constraints~\cite{ChekuriVondrakZenklusen2014}, prophet inequalities and Bayesian mechanism 
design~\cite{FSZ16,KW12}, constrained posted pricing mechanisms~\cite{FSZ16}, 
and stochastic probing~\cite{GN13,AdamczykWlodarczyk2018}. The CRS framework has also been extended 
to other settings, in particular to online~\cite{FSZ16} and random-order~\cite{AdamczykWlodarczyk2018} settings.
The quality of a CRS is generally measured by its \emph{balance}.

\begin{definition} Let $R(\mathbf{x})$ be a set of elements formed 
 by including each element $e$ 
independently with probability $x_e$.
A CRS $\pi$ is {\bf $(b,c)$-balanced} if 
for 
every $x \in bP$ and every $e \in \mathrm{support}(x)$, it is the case that $\Pr\bigl[e \in S \mid e \in R(\mathbf{x})\bigr] \;\geq\; c$ 
when $A=R(\mathbf{x})$. 
Here the probability is taken over the random events in the formation of $R(\mathbf{x})$, and potentially random events
internal to the CRS $\pi$. 
\end{definition}
Intuitively, a CRS being balanced means that
if $x$ is not near the upper boundary of $P$ and $A$ is randomly selected using the probabilities from $x$, 
then each element of $A$ survives $\pi$'s pruning with constant probability.

\subsection{Current Limits to the Applicability of CRS}

CRS techniques generally do not yield approximation guarantees for problems
with covering constraints, despite their ubiquity, because essentially all
CRS analyses provide only lower bounds on expected value. For example,
\cite{ChekuriVondrakZenklusen2014} shows that for any non-negative linear
objective $f(y) = \sum_e a_e y_e$, a $(b,c)$-balanced CRS outputs a set $S$
with $\mathbb{E}[f(S)] \geq c \cdot f(x)$; consequently,
$\mathbb{E}[|S|] \geq c \cdot \sum_{i=1}^n x_i$.

Concentration guarantees in the CRS literature are limited. Upper tails,
the probability that $f(S)$ exceeds a threshold, admit Chernoff-type bounds
because $R(\mathbf{x})$ is sampled independently and the CRS prunes it.
Lower tails, the probability that $f(S)$ falls below a threshold, are more
problematic, with essentially one known result. For a single matroid and
$f(S)=|S|$, the CRS of \cite{ChekuriVZ11} always outputs a set whose size
equals the rank of $R(\mathbf{x})$. A Chernoff-type lower-tail bound then
follows from the concentration theorem of \cite{swap_rounding_one_matroid}
for submodular functions over the unpruned random set $R(\mathbf{x})$.

{
The following example illustrates the need for lower-tail bounds when
using a CRS with covering constraints.\footnote{Independent rounding gives
better guarantees for this example; its role here is to illustrate the
concentration needed from a CRS.}
\par}

\smallskip
\noindent\emph{Representative Subhypergraph Problem:}
The input is a $k$-uniform $k$-partite hypergraph $H=(V,E)$ and parameters
$\alpha\le\beta<1$. A feasible solution is a subhypergraph $H'=(V,E')$
of $H$ in which every vertex $v$ has degree between $\alpha d(v)$ and
$\beta d(v)$, where $d(v)$ is its degree in $H$.

The natural feasibility constraints consist of $k$ collections of
partition matroid constraints (one per part) enforcing the upper degree
bounds, and covering constraints (one per vertex) enforcing the lower
degree bounds.

{
A CRS approach needs lower-tail bounds on $|S\cap E(v)|$, where $E(v)$
is the set of hyperedges incident to $v$. Acceptance decisions are
dependent, so standard Chernoff bounds do not directly apply. Moreover,
selections within $E(v)$ can be affected by random choices involving
hyperedges outside $E(v)$. Bounds strong enough to take a union bound
over vertices would ensure all covering constraints hold.
\par}

Thus the natural research question that we consider is:
\begin{center}
\emph{Can we extend contention resolution schemes to problems with covering
constraints?}
\end{center}
{
More concretely, is there a CRS for the intersection of $k$ matroids
such that, for every fixed target set $T$, $|S\cap T|$ attains nearly
$\sum_{e\in T}x_e/(k+1)$ with high probability when this benchmark is
large enough? In the example above, the target sets are the sets $E(v)$
of hyperedges incident to each vertex.
\par}

\subsection{Our Results}
{Our main contribution is lower-tail concentration for
arbitrary linear objectives with coefficients in $[0,1]$ under the
intersection of $k$ matroids. We establish this guarantee for
the random-order CRS of Adamczyk and
W{\l}odarczyk~\cite{AdamczykWlodarczyk2018}, which we call AW.
Taking $a_e=1$ for $e\in T$ and $a_e=0$ otherwise yields a lower-tail
bound on $|S\cap T|$ for \emph{every} fixed $T\subseteq U$.}

\begin{theorem}[Linear concentration for $k$ matroids]\label{thm:main-intro}
Let $M_1, \ldots, M_k$ be matroids on a common ground set $U$, and let 
$\mathbf{x}  \in \bigcap_{i=1}^k P(M_i)$ be a fractional point in the intersection 
of their {independence polytopes $P(M_i)$}. Let $a_e \in [0,1]$ for $e \in U$, and let $S$ be the output of AW, the $k$-matroid random-order CRS of~\cite{AdamczykWlodarczyk2018}. Let $Q = \sum_{e \in U} 
a_e x_e$ and $a(S) := \sum_{e\in S} a_e$. Then for all $\delta \in (0,1)$, $\Pr\!\left[a(S) \leq (1-\delta)\frac{Q}{k+1}\right]
    \leq \exp\!\left(-\frac{\delta^2 Q}{4(k+1)}\right).$
\end{theorem}

The bound is \emph{dimension-free}: the failure probability depends on
$Q$, $k$, and $\delta$, but not on $n=|U|$. As with the concentration
bounds for swap rounding~\cite{swap_rounding_one_matroid}, this is useful
when the fractional value $Q$ is much smaller than the ground-set size.
For fixed $\delta$, the failure probability decays exponentially in
$Q/(k+1)$.

The value $Q/(k+1)$ is the expectation guarantee established
in~\cite{AdamczykWlodarczyk2018}. Theorem~\ref{thm:main-intro} shows that
AW also attains any fixed fraction of this guaranteed value with high
probability when $Q/(k+1)$ is large. Note that the actual mean $\mathbb E[a(S)]$
could be larger than $Q/(k+1)$, but our theorem gives a bound on deviations below this guaranteed value.
{
Our analysis tracks both the value already collected and the fractional
value that remains available. The key additional ingredient beyond the
expectation analysis is a bound on how much total remaining fractional
value can be lost in one update. A general concentration theorem combines
this bound with a conditional comparison between expected loss and
expected gain. The resulting tail bound depends on the initial fractional
value rather than the number of updates. Conceptually, the proof also extends the pessimistic-estimator
method of Harvey and Olver for pipage rounding~\cite{HarveyOlver2014} to
AW's random-order process: bounded gain and expected-loss control replace
the swap-concavity argument. We view this extension as a contribution
of independent interest.
\par}

Our concentration guarantee also extends to monotone submodular functions.
For such a function $f$, its multilinear extension is $F(\mathbf{x})=\mathbb E[f(R(\mathbf{x}))]$,
where $R(\mathbf{x})$ includes each element $e\in U$ independently
with probability $x_e$.
The AW scheme guarantees expected value at least
$F(\mathbf{x})/(k+1)$~\cite{AdamczykWlodarczyk2018}.
We show that the same scheme achieves nearly this value with high
probability when this benchmark is large compared with the value
of any single element.
\begin{theorem}[Submodular concentration for $k$ matroids]
\label{thm:submodular-intro}
Let $M_1,\ldots,M_k$, $\mathbf{x}$, and $S$ be as in
\Cref{thm:main-intro}. Let $f:2^U\to\mathbb R_{\ge0}$ be monotone
and submodular with $f(\varnothing)=0$, and let $F$ be its multilinear
extension. Put $B=\max_{e\in U}f(\{e\})$ and assume $B>0$.
Then, for every $\delta\in(0,1)$, $\Pr\!\left[f(S)\leq(1-\delta)\frac{F(\mathbf{x})}{k+1}\right]
    \leq\exp\!\left(-\frac{\delta^2F(\mathbf{x})}{4(k+1)B}\right)$.
\end{theorem}

\medskip \noindent
\textbf{Applications:} Given the ubiquity of covering constraints, we expect that extending CRS to problems with both packing and covering constraints should significantly broaden their reach. 
To support this expectation, we apply  Theorem~\ref{thm:main-intro} to several natural problems involving both packing and covering constraints. To make this more accessible on a first read, 
we largely omit background definitions, and only state the results informally. Background definitions and formal statements of the results
can be found in the later technical sections.

\medskip  \noindent
\emph{Matroid Intersection Coloring.}
This problem consists of $k$ matroids $M_1, \ldots, M_k$ on a common ground set $U$.
A feasible solution is a partition of $U$ into parts that are each independent in all matroids. 
The objective is to minimize the number of parts. 
 The problem is NP-Hard, even for two matroids~\cite{BercziSchwarcz2021,gmpm_hard}. There are non-constructive existential results using topological fixed-point arguments~\cite{ab06,AharoniBergerGuoKotlar2025,BergerGuo2025}, and there are approximation  algorithms, some for special types of matroids~\cite{part_decomp_1,part_decomp_2,part_decomp_gammoid,arndt2025}. 
Let $\chimax=\max_{i\in[k]}\chi(M_i)$ be the largest chromatic number
of an individual matroid. In their recent FOCS paper, Arndt, Moseley,
Pruhs, Swamy, and Zlatin~\cite{ArndtMPS26} give a polynomial-time algorithm
using at most
$k(k-1)\chimax$ colors. They ask whether polynomial-time algorithms can
match or nearly match the nonconstructive bound of $(2k-1)\chimax$
colors~\cite[Question~1]{ArndtMPS26}.
We address this question when $\chimax$ is sufficiently large, obtaining
$O(k)\chimax$ colors and matching the dependence on $k$ up to a
constant factor.

\begin{theorem}\label{thm:coloring}
There is an absolute constant $C$ such that the following holds.
Let $M_1,\ldots,M_k$ be  matroids on a common ground set of
$n$ elements, and let
$\chimax=\max_{i\in[k]}\chi(M_i)$.
If $\chimax\geq Ck\log(2nk)$, then a randomized polynomial-time algorithm
returns a common coloring using $O(k)\chimax$ colors with probability
at least $1-n^{-10}$. In particular, this gives an $O(k)$-approximation.
\end{theorem}

To prove \Cref{thm:coloring}, we follow the stategy of~\cite{ArndtMPS26} by repeatedly removing common independent
sets to reduce the maximum remaining density, measured as set size
divided by rank.
The key ingredient there is the concentration guarantees of swap rounding, which only holds for the intersection of two matroids.
We replace swap rounding with the AW scheme and use our concentration
bound to extend this approach to $k$ matroids.

\medskip  \noindent
\emph{Monotone Submodular Maximization with Matroid, Packing, and Covering Constraints.} An instance of this problem consists of a monotone submodular function $f$ on a ground set $U$, $k$ matroids on $U$, packing constraints, and covering constraints.
A feasible solution is a subset of $U$ that is independent in all of the matroids, and satisfies the packing and covering constraints. The objective
is to maximize $f$. 
Monotone submodular maximization is a well-studied optimization problem receiving considerable attention over the past few decades \cite{FisherNW1978,NemhauserW1978,CalinescuChekuriPalVondrak2011,LeeSV2010}.  Under a single matroid constraint, the problem admits an optimal polynomial-time $(1-1/e)$-approximation via a continuous greedy algorithm, and pipage rounding or swap rounding~\cite{CalinescuChekuriPalVondrak2011}. Under a single matroid constraint and a constant number of packing and covering constraints, the problem also admits an optimal polynomial-time $(1-1/e-\eps)$-approximation for all fixed $\eps > 0$, allowing violation of the covering constraints by a factor of $1-\eps$ \cite{MizrachiSSU19}. Under the intersection of $k$ matroid constraints, but no packing and covering constraints, the best known polynomial-time approximation guarantee is $\frac{1}{k+\eps}$ for every fixed $\eps > 0$ \cite{LeeSV2010}, and the standard greedy algorithm gives a $\frac{1}{k+1}$-approximation \cite{FisherNW1978}.

As in~\cite{ChekuriVondrakZenklusen2014}, we work in what they called the ``loose constraints'' setting.
When all constraints are normalized so that the left-hand-side coefficients are in $[0,1]$, define the \textbf{capacity} of an instance to be the minimum right-hand-side value over the $p$ packing constraints, and the \textbf{width} of a covering instance to be the minimum right-hand-side value over the $c$ covering constraints.
 Theorem \ref{thm:main-intro} yields a bicriteria approximation result for monotone submodular maximization under $k$ matroid constraints,  packing constraints with logarithmic capacity, and covering constraints with width linear in $k$ and logarithmic in the number of covering constraints.

\begin{theorem}\label{thm:submax_main_temp}
Consider the problem of maximizing a monotone submodular function subject to $k$ matroid constraints, packing constraints of capacity at least $\Omega(\log(p))$, and covering constraints of width at least $\Omega(k\log(c))$. For this problem, there is a randomized polynomial-time algorithm that on instances where there is a feasible solution produces a solution that satisfies the matroid constraints, satisfies the packing constraints with high probability in $p$, violates the covering constraints by at most a factor of $O(k)$ with high probability in $c$, and
approximates the optimal objective
within a factor of $\left(\frac{1-1/e-\eps}{k+1}\right)$ in expectation for all fixed $\eps > 0$.
\end{theorem}

As far as we know, this is the first  approximation algorithm for monotone submodular maximization in the setting of $k$ matroid constraints, packing constraints, and covering constraints. As we shall explain later, the $O(k)$ violation of the covering constraints 
is  essentially best possible provided $NP \not \subseteq BPP$. 

Applying this result to the Representative Subhypergraph problem above
yields a randomized
polynomial-time algorithm that, on instances where there is a feasible solution,
and where the lower degree bounds are {at least a
sufficiently large constant times $k\log(|V|)$}, outputs a subhypergraph
$H'$ that satisfies the upper degree bounds,
and that
violates the lower degree bounds by at most an $O(k)$ factor (Appendix \ref{subsec:submax_hyper}).~\footnote{But again note that comparable or better results are obtainable by standard Chernoff-like tail bounds and union bounds.}

\medskip
\noindent\textbf{Roadmap to our algorithm and analysis.}
We develop the main ideas for a single matroid, where the blocking
and potential arguments are easiest to see. For $k$ matroids, we
apply the same blocking analysis within each matroid and combine
the resulting loss bounds in a single potential argument.

\smallskip
\noindent\emph{The AW process.}
Consider one matroid on a ground set $U$ with $n=|U|$.
Fix a feasible fractional point $\mathbf x$ and weights
$a_e\in[0,1]$, and write $Q=\sum_e a_ex_e$.
AW represents $\mathbf x$ as a convex combination of independent
sets. Each element with $x_e>0$ independently chooses a containing
set as its \emph{controller}, with probability proportional to
the set's weight.
AW independently includes each element $e$ in a sampled set $R$
with probability $x_e$ and processes elements in uniform random order.
A sampled arrival is accepted if it still belongs to its controller.
Exchanges insert each accepted element into every maintained set,
preserving independence and keeping all accepted elements in every set.
An unprocessed element becomes \emph{blocked} if an exchange removes
it from its controller; once blocked, it cannot be accepted.
Let $S$ be the final accepted set and $a(S)=\sum_{e\in S}a_e$.

\smallskip
\noindent\emph{Blocking gives balance.}
The AW blocking bound limits the conditional probability that an
update removes a fixed unprocessed element from a maintained set.
For fixed controller and target sets, injectivity of the exchange map
means that at most one element of the controller set can displace
a given target element.
Summing the resulting bounds gives
$\Pr(e\in S\mid e\in R)\geq1/2$ for each $e$ with $x_e>0$,
and hence $\mathbb E[a(S)]\geq Q/2$.
Section~\ref{sec:one-matroid-scheme} develops this argument.

\smallskip
\noindent\emph{Track gain and loss.}
After $t$ arrivals, let $V_t$ be the value collected so far.
Let $Z_t$ be the remaining fractional value: the sum of each
unprocessed element's weight times its conditional probability
of acceptance if it arrives next.
Define the one-step gain and loss by
\[
G_t=V_{t+1}-V_t,\qquad L_t=Z_t-Z_{t+1}.
\]
Conditional on \emph{any} fixed history, processing the arrival removes
as much residual value in expectation as we gain.
The blocking bound shows that, conditional on the history,
the expected loss caused by exchanges is at most the expected gain. Thus expected loss is at most twice expected gain.
Moreover, each maintained set loses at most one unprocessed element
per update, giving $0\leq G_t,L_t\leq1$.

\smallskip
\noindent\emph{Build a potential.}
Inspired by Harvey and Olver's pessimistic estimators for pipage
rounding~\cite{HarveyOlver2014}, we seek a nonnegative potential
whose conditional expectation does not increase, but whose final
value must greatly exceed its initial value when the output is
too small. These properties make a small output unlikely:
a nonnegative quantity cannot often be large if its expectation
remains small. We use
\[
\Phi_t=\exp(-\theta V_t-\eta Z_t),\qquad \theta,\eta>0.
\]
Collecting value lowers the potential; losing residual value raises it.
The bounds on individual updates let us bound the exponential
factors by linear functions of the gain and loss.
We can then choose the parameters so that the expected gain--loss
comparison makes the conditional expected change in the potential
nonpositive.

\smallskip
\noindent\emph{Derive the lower-tail bound.}
Initially $(V_0,Z_0)=(0,Q)$, while at termination
$(V_n,Z_n)=(a(S),0)$.
Thus $\Phi_0=e^{-\eta Q}$, and the event $a(S)\leq v$ implies
$\Phi_n\geq e^{-\theta v}$.
Since $\mathbb E[\Phi_n]\leq\Phi_0$, Markov's inequality gives the
following tail bound without requiring independence among acceptance
decisions:
\[
\Pr[a(S)\leq v]
\leq e^{\theta v}\mathbb E[\Phi_n]
\leq e^{\theta v-\eta Q}.
\]
Choosing the parameters yields concentration below $Q/2$.
Section~\ref{sec:gain-loss-concentration} formalizes this argument
as a general gain--loss concentration theorem.

\smallskip
\noindent\emph{Extend the argument.}
Section~\ref{sec:multi-matroid-concentration} applies the same theorem
to multiple matroids and normalized monotone submodular objectives.
The corresponding gain--loss bounds are proved in
Appendix~\ref{app:multi-matroid-proofs}.
Appendices~\ref{sec:matroidcoloring} and~\ref{sec:submodular}
present the applications to matroid intersection coloring and
constrained submodular maximization, together with their proofs.

{
\medskip \noindent 
\textbf{Organization.}
{
Section~\ref{sec:one-matroid-scheme} recalls one-matroid AW and establishes
a combinatorial bound on the expected removal of unprocessed elements
from maintained independent sets during updates.
Section~\ref{sec:gain-loss-concentration} translates this bound into
bounds on each step's gain and loss of remaining fractional value.
These yield a pessimistic estimator proving lower-tail concentration.
Section~\ref{sec:multi-matroid-concentration} extends the analysis to
multiple matroids and monotone submodular objectives.
Appendices~\ref{sec:matroidcoloring} and~\ref{sec:submodular} present
applications to matroid intersection coloring and constrained monotone
submodular maximization.
Appendix~\ref{app:pessimistic-estimators} compares our argument in detail
with Harvey and Olver's pessimistic-estimator technique~\cite{HarveyOlver2014}.
}
\par}

\section{One-matroid random-order contention resolution}
\label{sec:one-matroid-scheme}{
Our concentration proof begins with a combinatorial analysis of the
random-order contention resolution scheme AW. We develop the analysis for a single matroid, where the main ideas
behind our $k$-matroid concentration bounds are easiest to see.
For clarity of exposition, we state the general results in
Section~\ref{sec:multi-matroid-concentration} and defer their proofs
to Appendix~\ref{app:multi-matroid-proofs}. The algorithm maintains a collection of independent sets and
updates them by exchanges whenever it accepts an element. The key
property established in this section bounds the probability that such
an update removes a particular unprocessed element from a maintained
set. We also show how this property recovers AW's known $1/2$-balance
guarantee. In Section~\ref{sec:gain-loss-concentration}, we use it to
control the expected loss of remaining fractional value and construct
a pessimistic estimator that proves concentration.
\par}

We begin by describing the AW scheme. Let $M=(E,\mathcal I)$ be a matroid on $n$ elements, with
independence polytope $P(M)=\operatorname{conv}\{\mathbf1_B:B\in\mathcal I\}$.
Fix $x\in P(M)$ with a convex decomposition
\begin{equation}\label{eq:marek-matroid-decomposition}
x=\sum_{i=1}^{N}\beta_i\mathbf1_{B_i^0},\qquad
B_i^0\in\mathcal I,\quad \beta_i\ge0,\quad
\sum_{i=1}^{N}\beta_i=1.
\end{equation}
We retain labels when sets coincide. Thus the total support weight
containing $e$ is $x_e=\sum_{i:e\in B_i^0}\beta_i$. We use the standard matroid exchange-mapping property.
\begin{lemma}[Matroid exchange mapping]
\label{lem:marek-matroid-exchange-map}
For any $A,B\in\mathcal I$, there is a map
$\phi_{A,B}:A\to B\cup\{\bot\}$ that fixes $A\cap B$, is
injective on non-$\bot$ images, and satisfies $B-\phi_{A,B}(e)+e\in\mathcal I$ for $
(e\in A\setminus B), $
where $B-\bot:=B$. Such maps can be constructed in polynomial
time from an independence oracle
\cite[Corollary~39.12a]{schrijver_book}.
\end{lemma}

\paragraph{The AW scheme.}
The algorithm maintains independent supports $B_i^t$ with fixed
weights $\beta_i$. Each $e$ with $x_e>0$ independently draws a
\emph{controller} $c(e)$, {the label of a support initially
containing $e$}, according to
\begin{equation}\label{eq:aw-one-controller}
\Pr(c(e)=i)=\frac{\beta_i}{x_e}
\qquad(e\in B_i^0).
\end{equation}
These probabilities sum to one because the supports containing $e$
have total weight $x_e$.
{The label $c(e)$ stays fixed; $e$ remains eligible only
while it belongs to $B_{c(e)}^t$.}
Independently sample $R$ with $\Pr(e\in R)=x_e$ and draw a uniform
random order of $E$, independently of $R$ and the controllers.
Write $e_t$ for the element arriving at step $t$.

\begin{algorithm}[H]
\caption{AW random-order CRS for one matroid}
\label{alg:aw-one}
\begin{algorithmic}[1]
\Require Initial supports and exchange maps, sampled set $R$, and controllers $c(e)$
\State $S_0\gets\varnothing$
\For{$t=0,\ldots,n-1$}
  \If{$e_t\in R$ \textbf{and} $e_t\in B_{c(e_t)}^t$}
    \State $S_{t+1}\gets S_t\cup\{e_t\}$; $A\gets B_{c(e_t)}^t$
    \State Simultaneously set $B_i^{t+1}\gets B_i^t-\phi_{A,B_i^t}(e_t)+e_t$ for all $i$
    \State Recompute exchange maps
  \Else
    \State $S_{t+1}\gets S_t$; $B_i^{t+1}\gets B_i^t$ for all $i$
  \EndIf
\EndFor
\State \Return $S_n$
\end{algorithmic}
\end{algorithm}

Each exchange preserves independence and inserts the accepted
element into every support.
{Previously accepted elements belong to both supports and
map to themselves, so injectivity prevents their removal.}
Hence $S_t\subseteq B_i^t$ for all $i,t$,
so $S_n\subseteq R$ is independent.

Let $U_t$ be the unprocessed elements. The natural filtration
$\mathcal F_t$ reveals only the first $t$ arrivals and their sampling
bits and controllers, together with the current supports and exchange
maps. Thus, given $\mathcal F_t$, the next arrival is uniform in $U_t$,
and unrevealed bits and controllers remain independent with their
original laws. In particular, we never condition on all of $R$ or on
the final output. Updates never insert an unprocessed element into a
support, so $e\in U_t\cap B_i^t$ implies $e\in B_i^0$.

\paragraph{Blocking and balance.}
{
An unprocessed element becomes \emph{blocked} if it is removed from its
controller support before it arrives; it can then no longer be accepted.
The following bound controls the probability that an unprocessed element
is removed from any fixed current support in a single step.
\par}

\begin{lemma}[AW blocking bound]\label{lem:aw-blocking}
For every $f\in U_t$ and every current support $B_j^t\ni f$,
 $\Pr(f\notin B_j^{t+1}\mid\mathcal F_t)
\le \frac1{n-t}\sum_{i:f\notin B_i^t}\beta_i
\le \frac1{n-t}.$
\end{lemma}

\begin{proof}
Fix a target support $B_j^t\ni f$. For each controller support
$B_i^t$, {injectivity allows at most one unprocessed
element $e\ne f$ to map to $f$; if $f\in B_i^t$, no other element maps
to $f$ because $\phi_{B_i^t,B_j^t}(f)=f$.}
If a candidate $e$ exists, then $e\in U_t\cap B_i^0$. A label with
$\beta_i=0$ contributes nothing, so assume $\beta_i>0$ and hence
$x_e>0$. Deferred revelation gives $\Pr(e_t=e,\ e\in R,\ c(e)=i\mid\mathcal F_t)
=\frac1{n-t}x_e\frac{\beta_i}{x_e}
=\frac{\beta_i}{n-t}.$  Acceptance only lowers this probability. Summing over the labels
$i$ with $f\notin B_i^t$ proves the claimed bound.
\end{proof}

{
To pass from blocking probabilities to concentration, we track the
value already collected and the fractional value still available.
We define below $G_t$ be the gain in collected value and $L_t$ to
be the decrease in remaining fractional value. This loss includes
both the fractional value of the element being processed and any
value removed by exchanges.
\par}

\begin{lemma}[One-step gain--loss bound]
\label{lem:one-matroid-gain-loss}
Fix $a\in[0,1]^E$ and put $Q=a\cdot x$. For each round define
\begin{equation}\label{eq:one-matroid-potential}
\begin{gathered}
 x_e^t:=\sum_{i:e\in B_i^t}\beta_i,\qquad
V_t:=a(S_t),\qquad
Z_t:=\sum_{e\in U_t}a_ex_e^t,\\
G_t:=V_{t+1}-V_t,\qquad L_t:=Z_t-Z_{t+1}.
\end{gathered}
\end{equation}
Then $(V_0,Z_0)=(0,Q)$, $(V_n,Z_n)=(a(S),0)$, and almost surely,
for every $t<n$,
\[
0\le G_t,L_t\le1,\qquad
\E[L_t\mid\mathcal F_t]\le2\E[G_t\mid\mathcal F_t].
\]
In particular, $\sum_{t=0}^{n-1}L_t=Q$ on every execution.
\end{lemma}
\begin{proof}
Let $A_t:=a_{e_t}x_{e_t}^t$ be the loss from processing the arrival
and write the exchange loss as
$X_t:=\sum_{f\in U_t}a_f\mathbf1_{\{f\ne e_t\}}(x_f^t-x_f^{t+1})$,
so $L_t=A_t+X_t$. Uniform arrival and the controller probabilities give $\E[A_t\mid\mathcal F_t]=\E[G_t\mid\mathcal F_t]
=\frac{Z_t}{n-t}$. 

For each $f\in U_t$, the AW blocking bound gives $
\E[\mathbf1_{\{f\ne e_t\}}(x_f^t-x_f^{t+1})\mid\mathcal F_t]
\le\frac{x_f^t}{n-t}.$ Weighting by $a_f$ and summing yields
$\E[X_t\mid\mathcal F_t]\le Z_t/(n-t)$, hence the claimed
conditional gain--loss inequality.
{Within each maintained set, processing removes the arriving
element from its unprocessed portion if it is present; otherwise, an
exchange can remove at most one unprocessed element. These two losses
cannot both occur in the same set.}
Each support term loses at most one
unprocessed element in a round, so $G_t,L_t\in[0,1]$; the endpoint and
telescoping claims follow from the definitions.
\end{proof}

\begin{corollary}[One-half balance]\label{cor:aw-balance}
For every $e$ with $x_e>0$, the output $S=S_n$ satisfies $\Pr(e\in S\mid e\in R)\ge\frac12$.

\end{corollary}
\begin{proof}
Apply Lemma~\ref{lem:one-matroid-gain-loss} with $a=\mathbf1_{\{e\}}$.
Then $Q=x_e$, $V_n=\mathbf1_{\{e\in S\}}$, and $Z_n=0$, so $x_e=\E\sum_{t=0}^{n-1}L_t
\le2\E\sum_{t=0}^{n-1}G_t
=2\Pr(e\in S).$  Since $S\subseteq R$ and $\Pr(e\in R)=x_e$, this is equivalent to
$\Pr(e\in S\mid e\in R)\ge1/2$.
\end{proof}

{
The balance guarantee uses only the comparison of expected gain and
loss. To obtain concentration, we also use their one-step bounds.
The next section first illustrates this with a warmup, then proves a
bound independent of the number of steps.
\par}

\section{Concentration from bounded residual loss}
\label{sec:gain-loss-concentration}

\subsection{Warmup: a step-dependent concentration bound}

We begin with a warmup to build intuition. For any $a\in[0,1]^E$, let $a(T)=\sum_{e\in T}a_e$,
$Q=a\cdot x$, and $n=|E|$. As a warmup, we prove
 $\Pr\left(a(S)\le\frac Q2-u\right)
\le \exp\left(-\frac{u^2}{4n}\right),$ for $ 0<u\le Q/2.$
{
This preliminary bound introduces the exponential potential used in
our main argument: it combines collected and remaining fractional
value and does not increase in conditional expectation.
Its dependence on $n$ is too costly when $Q$ is small.
Indeed, setting $u=\delta Q/2$ for $\delta\in(0,1)$ gives
$\exp(-\delta^2Q^2/(16n))$, which can be close to one even when
$Q$ is of order $\log n$. After the warmup, we refine the potential
to obtain an exponent proportional to $Q$, allowing concentration
for covering requirements of logarithmic size.
\par}

Take the processes $(V_t,Z_t,G_t,L_t)$ from
Lemma~\ref{lem:one-matroid-gain-loss}. They start at
$(V_0,Z_0)=(0,Q)$ and end at $(V_n,Z_n)=(a(S),0)$, with
$V_{t+1}-V_t=G_t$ and $Z_t-Z_{t+1}=L_t$. The lemma says that the
gain, after charging half the residual loss, has nonnegative
conditional expectation and is bounded:
$
\E[G_t-L_t/2\mid\mathcal F_t]\ge0$, and $|G_t-L_t/2|\le1$.

Thus for $0<\lambda\le1$, applying $e^{-z}\le1-z+z^2$ to
$z=\lambda(G_t-L_t/2)$ and taking conditional expectations gives
\[
\begin{aligned}
\E[e^{-\lambda(G_t-L_t/2)}\mid\mathcal F_t]
&\le 1-\lambda\E[G_t-L_t/2\mid\mathcal F_t]\\
&\quad+\lambda^2\E[(G_t-L_t/2)^2\mid\mathcal F_t]
\le 1+\lambda^2\le e^{\lambda^2}.
\end{aligned}
\]
Now fix the target $v=Q/2-u$ and define
$\Phi_t:=e^{\lambda v+\lambda^2(n-t)}\exp(-\lambda V_t-\tfrac{\lambda}{2}Z_t)$.
The factor $e^{\lambda^2(n-t)}$ compensates for the one-step cost
$e^{\lambda^2}$ above. Using the gain and loss increments, we get
\begin{equation}\label{eq:warmup-potential-drift}
\E[\Phi_{t+1}\mid\mathcal F_t]
=\Phi_t e^{-\lambda^2}
  \E[e^{-\lambda G_t+\lambda L_t/2}\mid\mathcal F_t]
\le \Phi_t.
\end{equation}
At termination, $\Phi_n=e^{\lambda(v-a(S))}\ge1$ on the event
$a(S)\le v$, while initially
$\Phi_0=e^{\lambda v+\lambda^2n-\lambda Q/2}
=e^{-\lambda u+\lambda^2n}$. The supermartingale inequality
therefore gives $\Pr(a(S)\le Q/2-u)\le e^{-\lambda u+\lambda^2n}$  for $0<\lambda\le1$. Since $Q\le n$ and $0<u\le Q/2$,
choosing $\lambda=u/(2n)$ satisfies this condition (indeed,
$\lambda\le1/4$). With this choice,
$e^{-\lambda u+\lambda^2n}\le\exp(-u^2/(4n))$, proving
the claimed bound.

{
The factor $e^{\lambda^2(n-t)}$ pays the same error cost at every
remaining step, which introduces $n$ into the bound. We next adjust
the relative weights of $V_t$ and $Z_t$ so that the potential is a
supermartingale without this extra factor.
\par}

\subsection{Dimension-free concentration bound}
{
We state the sharper argument for general gain and loss processes,
so that the same theorem can later be applied to multiple matroids
and submodular objectives. The only inputs are a bound on expected
loss relative to expected gain and bounds on the size of each change.
\par}
\paragraph{The gain--loss interface.}
{Fix a filtration $(\mathcal F_t)_{t\ge0}$.}
For non-negative {$(\mathcal F_t)$-adapted} processes $(V_t,Z_t)$, set
\[
G_t:=V_{t+1}-V_t,
\qquad
L_t:=Z_t-Z_{t+1}.
\]
We say that they satisfy the \emph{gain--loss interface} with parameters
$(\alpha,\rho,B)$ if
\begin{equation}\label{eq:gain-loss-interface}
0\le G_t\le B,\qquad 0\le L_t\le\rho B,\qquad
\E[L_t\mid\mathcal F_t]\le\alpha\,\E[G_t\mid\mathcal F_t].
\end{equation}
Here $\alpha$ bounds expected loss per unit gain, while $B$
and $\rho B$ bound the one-step gain and loss. For one-matroid AW,
Lemma~\ref{lem:one-matroid-gain-loss} gives
$(\alpha,\rho,B)=(2,1,1)$.

{
The tail bound below depends on the initial fractional value $Q$
rather than the number of steps. We also retain the exponential-moment
bound, which lets us combine several AW runs in the coloring
application (Appendix~\ref{sec:matroidcoloring}).
\par}

\begin{theorem}[Concentration from bounded residual loss]
\label{thm:gain-loss-concentration}
Let $(V_t,Z_t)$ be nonnegative adapted processes with
$V_0=0$, $Z_0=Q\geq0$, and $Z_\tau=0$ at a bounded stopping
time $\tau$. Suppose \eqref{eq:gain-loss-interface} holds for
some $\alpha,\rho,B>0$ at every step before $\tau$.
Then $\E[V_\tau]\geq\frac Q\alpha.$ Moreover, for every $\theta>0$,
\begin{equation}\label{eq:gain-loss-moment}
\E\left[e^{-\theta V_\tau/B}\right]
\leq
\left(
1+\frac{\rho}{\alpha}(1-e^{-\theta})
\right)^{-Q/(\rho B)}.
\end{equation}
Consequently, for every $\delta\in(0,1)$,
\begin{equation}\label{eq:gain-loss-tail}
\Pr\left(V_\tau\leq(1-\delta)\frac Q\alpha\right)
\leq\exp\left(-\frac Q B J_{\alpha,\rho}(\delta)\right)
\leq\exp\left(-\frac{\delta^2Q}{2B(\alpha+\rho)}\right),
\end{equation}
where $J_{\alpha,\rho}(\delta)=\frac1\rho\log\frac{\alpha+\rho}{\alpha+\rho(1-\delta)}+\frac{1-\delta}{\alpha}\log\frac{(1-\delta)(\alpha+\rho)}{\alpha+\rho(1-\delta)}.$
\end{theorem}

\begin{proof}
Since $Z_0-Z_\tau=Q$, summation gives 
$Q=\E\sum_{t<\tau}L_t
\leq\alpha\E\sum_{t<\tau}G_t
=\alpha\E[V_\tau].$  
For $\theta>0$, {consider
$\Phi_t=\exp(-[\theta V_t+\eta Z_t]/B)$,
where $\eta>0$ will be chosen below.
Collecting value decreases this potential, while losing remaining
fractional value increases it.
We choose $\eta$ so that, despite these losses, the potential does
not increase in conditional expectation.}
{Set} $c=1-e^{-\theta}$ and
$\eta=\rho^{-1}\log(1+\rho c/\alpha)$.
{Convexity on the bounded gain and loss intervals gives}
\[
e^{-\theta G_t/B+\eta L_t/B}\leq\left(1-\frac{cG_t}{B}\right)\left(1+\frac{cL_t}{\alpha B}\right)\leq1-\frac{cG_t}{B}+\frac{cL_t}{\alpha B}.
\]
Taking conditional expectations and using the interface shows
$\E[e^{-\theta G_t/B+\eta L_t/B}\mid\mathcal F_t]\le1$.
{Hence $\Phi_t$ is a nonnegative}
supermartingale up to $\tau$. Optional stopping and the endpoint
values give $\E\left[e^{-\theta V_\tau/B}\right]\leq\E[\Phi_\tau]\leq\Phi_0=e^{-\eta Q/B}=\left(1+\frac{\rho}{\alpha}(1-e^{-\theta})\right)^{-Q/(\rho B)}. $ This proves \eqref{eq:gain-loss-moment}.

For $Q>0$, Markov's inequality and the moment bound give, for
$\theta>0$, $\Pr\left(V_\tau\leq(1-\delta)\frac Q\alpha\right)
\leq\exp\left[\frac QB\left(\frac{(1-\delta)\theta}{\alpha}
-\frac1\rho\log\left(1+\frac{\rho}{\alpha}(1-e^{-\theta})\right)\right)\right].$ The exponent is minimized at
$\theta_*=\log\frac{\alpha+\rho(1-\delta)}{(1-\delta)(\alpha+\rho)}>0$;
substitution yields the first bound in~\eqref{eq:gain-loss-tail}.

Finally, $J_{\alpha,\rho}(0)=J_{\alpha,\rho}'(0)=0$, while  $J_{\alpha,\rho}''(\delta)=\frac1{(1-\delta)[\alpha+\rho(1-\delta)]}\geq\frac1{\alpha+\rho}.$  Integrating twice gives $J_{\alpha,\rho}(\delta)\geq\delta^2/[2(\alpha+\rho)]$.
For $Q=0$ the claim is immediate.
\end{proof}

\begin{corollary}
\label{cor:one-matroid-lower-tail}
For one-matroid AW, $a\in[0,1]^E$, and
$Q=\sum_{e\in E}a_ex_e$, every $\delta\in(0,1)$ satisfies
\[\Pr\left(
a(S)\leq(1-\delta)\frac Q2
\right)
\leq
\exp\left(-QJ_{2,1}(\delta)\right)
\leq
\exp\left(-\frac{\delta^2Q}{6}\right).\]
\end{corollary}

\begin{proof}
Apply Theorem~\ref{thm:gain-loss-concentration} with
$(\alpha,\rho,B)=(2,1,1)$ from
Lemma~\ref{lem:one-matroid-gain-loss}.
\end{proof}

\section{Concentration for multiple matroids and submodular objectives}
\label{sec:multi-matroid-concentration}

{
We apply Theorem~\ref{thm:gain-loss-concentration} to the $k$-matroid
AW scheme, first for linear objectives and then for monotone submodular
objectives. In each case, we define the remaining fractional value and
state the bounds on one-step gain and loss that yield concentration.
Appendix~\ref{app:multi-matroid-proofs} describes the scheme and proves
these bounds.
\par}

\medskip  
\noindent \textbf{Linear objectives}
\label{sec:multi-linear} {
Let $M_1,\ldots,M_k$ be matroids on $E$, put $n=|E|$, and fix
$x\in\bigcap_{j=1}^kP(M_j)$. For their AW run, use $S_t$, $U_t$, and
$\mathcal F_t$ as in Section~\ref{sec:one-matroid-scheme}, and write
$S=S_n$. For $e\in U_t$, let $q_e^t$ be its probability of acceptance
if it arrives next, conditional on $\mathcal F_t$; set $q_e^t=0$ for
processed elements. Write $q^t=(q_e^t)_{e\in E}$.
{For $k=1$ and $e\in U_t$, this agrees with the earlier
notation: $q_e^t=x_e^t$.}
\par}

Fix $a\in[0,1]^E$, put $Q=a\cdot x$, and set $V_t:=\sum_{e\in S_t}a_e$, $Z_t:=\sum_{e\in U_t}a_eq_e^t,$ $G_t:=V_{t+1}-V_t,\quad L_t:=Z_t-Z_{t+1}.$ Then $(V_0,Z_0)=(0,Q)$ and $(V_n,Z_n)=(a(S),0)$.
{Appendix~\ref{app:multi-matroid-proofs} proves the gain--loss bounds}
\begin{equation}\label{eq:multi-linear-interface}
0\le G_t\le1,\qquad 0\le L_t\le k+1,\qquad
\mathbb E[L_t\mid\mathcal F_t]
\le(k+1)\mathbb E[G_t\mid\mathcal F_t],
\qquad \sum_tL_t=Q.
\end{equation}

\begin{theorem}[Linear concentration for $k$ matroids]
\label{thm:multi-linear}
For every $\delta\in(0,1)$,
 $\Pr\!\left[a(S)\le(1-\delta)\frac{Q}{k+1}\right]
\le\exp\!\left(-\frac{Q}{k+1}I(\delta)\right)
\le\exp\!\left(-\frac{\delta^2Q}{4(k+1)}\right),$ 
where $I(\delta):=(1-\delta)\log(1-\delta)
-(2-\delta)\log(1-\delta/2). $ The bound is independent of $|E|$.
\end{theorem}

\begin{proof}
Apply Theorem~\ref{thm:gain-loss-concentration} to
\eqref{eq:multi-linear-interface} with
$(\alpha,\rho,B)=(k+1,k+1,1)$ and use
$J_{k+1,k+1}(\delta)=I(\delta)/(k+1)$.
\end{proof}

For $k=1$, Section~\ref{sec:gain-loss-concentration} gives the
sharper exponent $-\delta^2Q/6$.

\medskip  
\noindent \textbf{Monotone submodular objectives}
\label{sec:multi-submodular} Let $f:2^E\to\mathbb R_{\ge0}$ be normalized, monotone, and
submodular, with multilinear extension $F(y)=\mathbb E[f(R(y))]$.
Let $B:=\max_{e\in E}f(\{e\})>0$; otherwise $f\equiv0$. Define $V_t:=f(S_t),$  $Z_t:=F(\mathbf1_{S_t}+q^t)-f(S_t),$  $G_t:=V_{t+1}-V_t,\quad L_t:=Z_t-Z_{t+1}.${
Conditional on the current history, $Z_t$ is the expected additional
value from augmenting $S_t$ with an independent auxiliary sample having
marginals $q^t$. Since $q^t$ vanishes on processed elements,
$\mathbf1_{S_t}+q^t\in[0,1]^E$.
\par}
Then $(V_0,Z_0)=(0,F(x))$ and $(V_n,Z_n)=(f(S),0)$.
{Appendix~\ref{app:multi-matroid-proofs} proves}
\begin{equation}\label{eq:multi-submodular-interface}
\begin{aligned}
0&\le G_t\le B, & 0&\le L_t\le(k+1)B,\\
\mathbb E[L_t\mid\mathcal F_t]
&\le(k+1)\mathbb E[G_t\mid\mathcal F_t],
&\qquad \sum_tL_t=F(x).
\end{aligned}
\end{equation}

\begin{theorem}[Submodular concentration for $k$ matroids]
\label{thm:multi-submodular}
For every $\delta\in(0,1)$, $\Pr\!\left[f(S)\le(1-\delta)\frac{F(x)}{k+1}\right]
\le\exp\!\left(-\frac{F(x)}{(k+1)B}I(\delta)\right)
\le\exp\!\left(-\frac{\delta^2F(x)}{4(k+1)B}\right),$ 
with $I$ as in Theorem~\ref{thm:multi-linear}.
\end{theorem}

\begin{proof}
Apply Theorem~\ref{thm:gain-loss-concentration} to
\eqref{eq:multi-submodular-interface} with
$(\alpha,\rho,B)=(k+1,k+1,B)$.
\end{proof}

\paragraph{Applications}
Appendix~\ref{sec:matroidcoloring} uses the CRS to provide an algorithm for matroid
intersection coloring and proves its approximation guarantee.
Appendix~\ref{sec:submodular} presents our algorithm for monotone
submodular maximization under matroid, packing, and covering constraints
and proves the objective and constraint guarantees stated in the introduction.

\section{AI Disclosure}

An earlier version of this work was developed without AI assistance and established concentration bounds through a direct analysis of the underlying random process using Freedman’s inequality. In preparing the present version, the authors used generative AI tools to help improve and simplify the analysis through the use of pessimistic estimators, and to assist with writing and editing. All proofs in the manuscript were rewritten and verified by the authors. The authors take full responsibility for the correctness, originality, and integrity of the manuscript.

\bibliographystyle{plain}
\bibliography{bib}

\appendix
\renewcommand*{\theHsection}{appendix.\Alph{section}}

\section{Comparison to the Pessimistic Estimator Framework}
We compare our methods with pessimistic estimator framework of prior work.
Indeed, our proof uses the same supermartingale structure as the
pessimistic-estimator method of Harvey and Olver~\cite{HarveyOlver2014}.
For product rounding with marginals $y$, the Chernoff estimator for
$a(R)\le v$ is
\[
q_{v,\theta}(y)=e^{\theta v}\prod_{e\in E}
  \bigl(1-y_e(1-e^{-\theta a_e})\bigr)
=\E[e^{\theta(v-a(R))}].
\]
It is at least one at an integral point $y=\mathbf1_X$ whenever
$a(X)\le v$. Harvey and Olver show that $q_{v,\theta}$ is concave
along swaps. If $y=\lambda y^-+(1-\lambda)y^+$ are the feasible
endpoints, then
\[
\min\{q_{v,\theta}(y^-),q_{v,\theta}(y^+)\}
\le\lambda q_{v,\theta}(y^-)+(1-\lambda)q_{v,\theta}(y^+)
\le q_{v,\theta}(y).
\]
Randomized pipage samples an endpoint with these mean-preserving
probabilities, giving $\E[q_{v,\theta}(y^{t+1})\mid y^t]
\le q_{v,\theta}(y^t)$; deterministic pipage chooses the endpoint
with smaller $q$.

Our AW proof uses the potential
\[
g_t=e^{\theta v}\exp(-\theta V_t-\eta Z_t),\qquad
\eta=\log\left(1+\frac{1-e^{-\theta}}2\right).
\]
AW is not a two-endpoint swap update, so we do not obtain this
minimum-endpoint guarantee for $q_{v,\theta}$. We need only the
conditional-expectation bound for our own estimator:
$\E[g_{t+1}\mid\mathcal F_t]\le g_t$. Since $g_n\ge1$ when
$a(S)\le v$, this suffices for the tail bound. Thus both randomized
methods make a pessimistic estimator a supermartingale: pipage uses
swap concavity, while AW uses bounded expected residual loss.

The next section verifies this interface for $k$ matroids
and monotone submodular objectives.

\section{Proofs of the multi-matroid gain--loss bounds}
\label{app:multi-matroid-proofs}

We prove the gain--loss bounds stated in
Section~\ref{sec:multi-matroid-concentration}. Throughout, $q^t$ is
the product residual from~\eqref{eq:multi-product-residual}; by
\eqref{eq:multi-acceptance}, it is the conditional probability that
an unprocessed element is accepted when it arrives.

\subsection{AW scheme and product residual}

Let $M_1,\ldots,M_k$ be matroids on $E$, $|E|=n$, and
$x\in\bigcap_{j=1}^kP(M_j)$. For each $j$, fix a decomposition
\[
x=\sum_b\beta_b^j\mathbf1_{B_b^{j,0}},\qquad
\sum_b\beta_b^j=1,\qquad B_b^{j,0}\in\mathcal I(M_j).
\]
Use one random order and independently sample $R$ with
$\Pr(e\in R)=x_e$. Independently for each $e$ and $j$, choose
controller $c_j(e)=b$ with probability $\beta_b^j/x_e$ among
supports containing $e$. Accept an arrival iff it lies in $R$ and
all its controllers remain live; update each matroid by exchange.
Use $U_t$ and $\mathcal F_t$ as in
Section~\ref{sec:one-matroid-scheme}.

For $e\in U_t$ and matroid $j$, define
\[
w_{e,j}^t:=\sum_{b:e\in B_b^{j,t}}\beta_b^j,
\qquad u_{e,j}^t:=w_{e,j}^t/x_e\quad(x_e>0),
\]
and the product residual
\begin{equation}\label{eq:multi-product-residual}
q_e^t:=
\begin{cases}
x_e\prod_{j=1}^k u_{e,j}^t,&e\in U_t,\ x_e>0,\\
0,&\text{otherwise}.
\end{cases}
\end{equation}
Deferred revelation gives
\begin{equation}\label{eq:multi-acceptance}
\Pr(e\text{ is accepted}\mid\mathcal F_t,e_t=e)=q_e^t.
\end{equation}

\subsection{Product-residual exchange loss}

For unprocessed $f$, define the mass lost from matroid-$j$ support
sets at step $t$ by
\[
D_{f,j}^t:=\sum_{b:f\in B_b^{j,t}}
\beta_b^j\mathbf1_{\{f\notin B_b^{j,t+1}\}},
\qquad
\Delta_f^t:=\begin{cases}
q_f^t-q_f^{t+1},&f\ne e_t,\\
0,&f=e_t.
\end{cases}
\]

\begin{lemma}[Product-residual exchange loss]
\label{lem:multi-exchange-loss}
For every $f\in U_t$ and $j\in[k]$,
\begin{equation}\label{eq:multi-single-loss}
\mathbb E[D_{f,j}^t\mid\mathcal F_t]
\le\frac{w_{f,j}^t(1-w_{f,j}^t)}{n-t}.
\end{equation}
Consequently,
\begin{equation}\label{eq:multi-product-loss}
\mathbb E[\Delta_f^t\mid\mathcal F_t]
\le\frac{kq_f^t}{n-t}.
\end{equation}
On acceptance of $e_t=e$,
\begin{equation}\label{eq:multi-pathwise-loss}
\sum_{f\in U_t\setminus\{e\}}\Delta_f^t\le k;
\end{equation}
on rejection, all exchange losses vanish.
\end{lemma}

\begin{proof}
Fix $j$, a target support $B_b^{j,t}\ni f$, and a controller
label $a$ with $f\notin B_a^{j,t}$. Injectivity of the exchange map
allows at most one candidate element to remove $f$ from this target
using controller $a$. If that candidate is $e$, it must arrive next,
belong to $R$, and choose $a$ as its controller in matroid $j$.
Deferred revelation therefore bounds the probability of this event by
\[
\frac{1}{n-t}x_e\frac{\beta_a^j}{x_e}
=\frac{\beta_a^j}{n-t}.
\]
The other $k-1$ controller tests can only reduce the probability that
the exchange is actually made. Summing over targets and controller
labels therefore gives
\[
\begin{aligned}
\mathbb E[D_{f,j}^t\mid\mathcal F_t]
&\le\frac1{n-t}\sum_{b:f\in B_b^{j,t}}\beta_b^j
  \sum_{a:f\notin B_a^{j,t}}\beta_a^j\\
&=\frac{w_{f,j}^t(1-w_{f,j}^t)}{n-t},
\end{aligned}
\]
proving~\eqref{eq:multi-single-loss}.

For $f\ne e_t$, each $u_{f,j}^t$ only decreases. The product
difference inequality on $[0,1]^k$ gives
\begin{equation}\label{eq:multi-product-difference}
0\le\Delta_f^t\le
\sum_{j=1}^kD_{f,j}^t\prod_{h\ne j}u_{f,h}^t.
\end{equation}
For $x_f>0$, taking conditional expectations and applying
\eqref{eq:multi-single-loss} yields
\[
\begin{aligned}
\mathbb E[\Delta_f^t\mid\mathcal F_t]
&\le\frac1{n-t}\sum_{j=1}^k
w_{f,j}^t(1-w_{f,j}^t)\prod_{h\ne j}u_{f,h}^t\\
&=\frac{q_f^t}{n-t}\sum_{j=1}^k(1-w_{f,j}^t)
\le\frac{kq_f^t}{n-t}.
\end{aligned}
\]
For $x_f=0$, all terms vanish, proving~\eqref{eq:multi-product-loss}.

On acceptance of $e$, each matroid-$j$ exchange removes at most
one element from each support not already containing $e$. Therefore
\[
\sum_{f\in U_t\setminus\{e\}}D_{f,j}^t\le1-w_{e,j}^t.
\]
Summing~\eqref{eq:multi-product-difference} over $f\ne e$, using
$u_{f,h}^t\le1$, and then summing over $j$ proves
\eqref{eq:multi-pathwise-loss}. Rejection triggers no exchanges.
\end{proof}

\subsection{Linear gain--loss bound}

\begin{proof}[Proof of the gain--loss bounds in~\eqref{eq:multi-linear-interface}]
For $e=e_t$, split the loss into arrival and exchange parts:
\[
L_t^{\mathrm{arr}}:=a_eq_e^t,\qquad
L_t^{\mathrm{ex}}:=\sum_{f\in U_t\setminus\{e\}}a_f\Delta_f^t,
\qquad L_t=L_t^{\mathrm{arr}}+L_t^{\mathrm{ex}}.
\]
By uniform arrival and~\eqref{eq:multi-acceptance},
\[
\mathbb E[L_t^{\mathrm{arr}}\mid\mathcal F_t]
=\mathbb E[G_t\mid\mathcal F_t]
=\frac1{n-t}\sum_{f\in U_t}a_fq_f^t.
\]
Lemma~\ref{lem:multi-exchange-loss} bounds the exchange part by
\[
\mathbb E[L_t^{\mathrm{ex}}\mid\mathcal F_t]
\le\frac{k}{n-t}\sum_{f\in U_t}a_fq_f^t
=k\mathbb E[G_t\mid\mathcal F_t].
\]
Thus $\mathbb E[L_t\mid\mathcal F_t]\le
(k+1)\mathbb E[G_t\mid\mathcal F_t]$. Also $0\le G_t\le1$.
The arrival loss is at most one; exchange loss is at most $k$ on
acceptance by~\eqref{eq:multi-pathwise-loss} and zero on rejection.
Hence $0\le L_t\le k+1$. Finally, $\sum_tL_t=Q$ by telescoping.
\end{proof}

\subsection{Submodular gain--loss bound}

\begin{proof}[Proof of the gain--loss bounds in~\eqref{eq:multi-submodular-interface}]
Condition on $\mathcal F_t$ and let $e=e_t$. For each $i\in U_t$,
set
\[
m_i:=f(S_t+i)-f(S_t),\qquad 0\le m_i\le B.
\]
Let $\bar q=q^t-q_e^t\mathbf1_e$ and
\[
h_e:=\partial_eF(\mathbf1_{S_t}+\bar q),
\qquad 0\le h_e\le m_e,
\]
where the upper bound follows from diminishing returns. Let
$S'_t=S_t+e$ on acceptance and $S'_t=S_t$ otherwise. Split
$L_t=L_t^{\mathrm{arr}}+L_t^{\mathrm{ex}}$ as
\[
\begin{aligned}
L_t^{\mathrm{arr}}
&:=Z_t-\bigl(F(\mathbf1_{S'_t}+\bar q)-f(S'_t)\bigr),\\
L_t^{\mathrm{ex}}
&:=F(\mathbf1_{S'_t}+\bar q)
-F(\mathbf1_{S'_t}+q^{t+1}).
\end{aligned}
\]
Since $F$ is affine in each coordinate,
\[
L_t^{\mathrm{arr}}=
\begin{cases}
q_e^th_e,&e\text{ rejected},\\
m_e-(1-q_e^t)h_e,&e\text{ accepted}.
\end{cases}
\]
Both cases lie in $[0,B]$. As the acceptance probability is $q_e^t$,
\[
\begin{aligned}
\mathbb E[L_t^{\mathrm{arr}}\mid\mathcal F_t,e_t=e]
&=(1-q_e^t)q_e^th_e
+q_e^t\bigl(m_e-(1-q_e^t)h_e\bigr)\\
&=q_e^tm_e
=\mathbb E[G_t\mid\mathcal F_t,e_t=e].
\end{aligned}
\]
Consequently,
\begin{equation}\label{eq:submod-gain}
\mathbb E[L_t^{\mathrm{arr}}\mid\mathcal F_t]
=\mathbb E[G_t\mid\mathcal F_t]
=\frac1{n-t}\sum_{i\in U_t}q_i^tm_i.
\end{equation}

For the exchange term, $q^{t+1}\le\bar q$ coordinatewise.
Submodularity bounds the decrease in $F$ from reducing coordinate
$i$ by $\Delta_i^t$ by $m_i\Delta_i^t$. This holds on acceptance
as well, since $S'_t\supseteq S_t$. Hence
\begin{equation}\label{eq:submod-exchange-domination}
0\le L_t^{\mathrm{ex}}
\le\sum_{i\in U_t\setminus\{e_t\}}m_i\Delta_i^t.
\end{equation}
By Lemma~\ref{lem:multi-exchange-loss} and~\eqref{eq:submod-gain},
\[
\mathbb E[L_t^{\mathrm{ex}}\mid\mathcal F_t]
\le\frac{k}{n-t}\sum_{i\in U_t}m_iq_i^t
=k\mathbb E[G_t\mid\mathcal F_t].
\]
Thus $\mathbb E[L_t\mid\mathcal F_t]
\le(k+1)\mathbb E[G_t\mid\mathcal F_t]$ and $0\le G_t\le B$.
Also $0\le L_t^{\mathrm{arr}}\le B$. On acceptance,
\eqref{eq:multi-pathwise-loss} and
\eqref{eq:submod-exchange-domination} give $L_t^{\mathrm{ex}}\le kB$;
on rejection it is zero. Therefore $0\le L_t\le(k+1)B$.
Telescoping yields $\sum_tL_t=F(x)$.
\end{proof}

\section{Application to Matroid Intersection Coloring}
\label{sec:matroidcoloring}

{
We apply the concentration analysis of
Section~\ref{sec:multi-matroid-concentration} to matroid intersection
coloring, obtaining an $O(k)$-approximation when the maximum individual
chromatic number is sufficiently large. Each run of AW produces a
common independent set whose previously uncolored elements form a new
color class. We use repeated runs to reduce the number of colors needed
for the remaining elements, then finish with an existing coloring
algorithm. To ensure sufficient progress across several runs, we combine
the exponential-moment bounds supplied by our concentration proof.
\par}

Let $M_1,\ldots,M_k$ be loopless matroids on a common ground set $U$ of
size $n$, and write $r_j$ for the rank function of $M_j$.
A coloring partitions $U$ into sets that are independent in every
$M_j$. Write $\Mint$ for this common-independent-set system.
Let $\OPT=\chi(\Mint)$ be the minimum number of colors, and let
$\chimax=\max_{j\in[k]}\chi(M_j)$. Since a common coloring is also a
coloring of each individual matroid, $\chimax\leq\OPT$.
A matroid with a loop admits no coloring, so we restrict attention to
loopless inputs. A single matroid can be colored optimally in polynomial
time~\cite{schrijver_book}; below we assume $k\geq2$ and $n\geq2$.

\begin{theorem}
\label{thm:coloring-crs}
There is an absolute constant $C$ such that, whenever
$\chimax\geq Ck\log(2nk)$, a randomized polynomial-time algorithm
returns a feasible coloring of $\Mint$ using $O(k)\chimax$ colors with
probability at least $1-n^{-10}$. In particular, it is an
$O(k)$-approximation for matroid intersection coloring.
\end{theorem}

{
The analysis uses the following exponential-moment bound to control
progress across a batch of AW runs. It follows from the general moment
bound in Section~\ref{sec:gain-loss-concentration} and the gain--loss
bounds verified in the proof of Theorem~\ref{thm:multi-linear}.
\par}

\begin{corollary}[Exponential-moment bound]
\label{cor:cr-aw-transform}
Let $x\in\bigcap_{j=1}^k P(M_j)$, and let $S$ be
the output of the $k$-matroid AW scheme.
For any {$a\in[0,1]^U$}, let $Q=a\cdot x$.
Then, for every $\theta>0$,
\begin{equation}\label{eq:cr-aw-transform}
\mathbb E[e^{-\theta a(S)}]
\leq (2-e^{-\theta})^{-Q/(k+1)}.
\end{equation}
\end{corollary}

\begin{proof}
The proof of \Cref{thm:multi-linear} establishes
that the gain--loss interface holds with
$(\alpha,\rho,B)=(k+1,k+1,1)$.
{Applying \eqref{eq:gain-loss-moment} gives}
\[
\begin{aligned}
\mathbb E[e^{-\theta a(S)}]
&\leq
\left(1+\frac{\rho}{\alpha}
(1-e^{-\theta})\right)^{-Q/(\rho B)}\\
&=(2-e^{-\theta})^{-Q/(k+1)}.
\end{aligned}
\]
\end{proof}

The strategy of repeatedly removing rounded independent sets and then
coloring a small residual instance appears in the two-matroid algorithm
of Arndt, Moseley, Pruhs, Swamy, and Zlatin~\cite{ArndtMPS26}.
We use their polynomial-time algorithm that colors any $k$-matroid
instance with at most $k(k-1)$ times its maximum individual chromatic
number~\cite[Theorem~1.4(b)]{ArndtMPS26} to finish the residual instance.
Our contribution here is the density reduction for $k$ matroids:
combining the CRS moment bound across a batch of samples controls every
flat.

\subsection{Overview}

The obstacle to using few colors is a set with many elements relative
to its rank. A color class contains at most $r_j(A)$ elements of any
set $A$, so coloring $A$ requires at least $|A|/r_j(A)$ colors.
This ratio is the \emph{density} of $A$ in $M_j$.
For the uncolored set $W\subseteq U$, let $M_j|_W$ denote the restriction
of $M_j$ to $W$, and define
\[
    \rho(W)=\max_{j\in[k]}\max_{\varnothing\ne A\subseteq W}
                  \frac{|A|}{r_j(A)},\qquad
    q(W)=\max_{j\in[k]}\chi(M_j|_W).
\]
Set $\rho(\varnothing)=q(\varnothing)=0$.
Edmonds' matroid coloring formula gives $q(W)=\lceil\rho(W)\rceil$. These individual chromatic numbers can be computed by matroid partition
algorithms~\cite{schrijver_book}.
Moreover, if $q\geq\rho(W)$ and $q\geq1$, the uniform point
$\mathbf x=\mathbf1_W/q$ belongs to the independence polytope $P(M_j|_W)$ for all $j$. To see this, note that for every
$A\subseteq W$, its mass $|A|/q$ is at most $r_j(A)$.
Thus it is a valid input to the AW-CRS scheme.

\paragraph{Reducing the density.}
Suppose the current value of $q(W)$ is $q$. We run a batch of fresh CRS
samples at $\mathbf1_W/q$ and use them as color classes.
If $A$ is any set of elements still uncovered within the batch, the
expected number selected from $A$ by the next sample is at least
$|A|/((k+1)q)$. This suggests spending $O(kq)$ colors to cover a
constant fraction of each dense set. In particular, in our key lemma (\Cref{lem:coloring-density-halving}), we prove that after $\lceil8(k+1)q\rceil$ samples, the maximum residual density is
at most $q/2$ with high probability, provided $q\geq20\log(2nk)$.

\paragraph{Controlling every dense set.} To prove \Cref{lem:coloring-density-halving}, we must argue that the density of \textit{every} subset of elements drops sufficiently. However, note that it suffices to
consider flats: taking a set's closure preserves its rank and can only
add elements. By definition of $q$, a rank-$m$ flat initially contains at most $qm$ elements.
As long as more than $qm/2$ of them remain uncovered, there is enough
fractional mass for the next CRS sample to make progress.
Combining the exponential-moment bounds over the batch gives probability
at most $e^{-qm}$ that the flat still exceeds this threshold at the end.
There are at most $n^m$ rank-$m$ flats in each matroid, so
$q=\Omega(\log(nk))$ makes the union bound work for every rank. 

\paragraph{Counting the colors.}
Call a batch of CRS samples ``successful" if after removing the sampled sets, the density of the remaining elements in each matroid does indeed drop by a factor of $1/2$. Then after a succesful batch, the residual individual chromatic numbers of each matroid are
at most $\lceil q/2\rceil$. We recompute their maximum and then use fewer
samples in the next phase. Up to rounding, the overall cost of all batches of CRS samples is:
\[
    O(k)\left(\chimax+\frac{\chimax}{2}
                    +\frac{\chimax}{4}+\cdots\right)
    =O(k\chimax).
\]
We stop when the residual individual chromatic numbers are at most
$\lceil\chimax/k\rceil$. Once this point is reached, the existing $k(k-1)$-approximate coloring guarantee then
finishes the instance with another $O(k\chimax)$ colors.
The assumption $\chimax\geq Ck\log(2nk)$ ensures that every sampling
phase before this stopping point has a large enough $q$ for the
density-reduction lemma.

\subsection{The algorithm}

Algorithm~\ref{alg:coloring-halving} implements the density-reduction
strategy above. Each phase fixes the current uncolored set $W$ and
samples independently on its restricted matroids. Samples can overlap.
We give a new color only
to the previously uncolored elements of each sample (taking a subset
preserves independence, so this is a valid coloring). At the end of the phase, we recompute $q(W)$
for the elements still uncolored.

Set $h=\lceil\chimax/k\rceil$, where $\chimax$ is the maximum individual
chromatic number of the original instance. This is the stopping threshold:
once $q(W)\leq h$, the final coloring algorithm uses at most $k(k-1)h$
additional colors, which is $O(k\chimax)$ under our assumption.
We allow at most $\ell=\lceil\log_2 k\rceil$ sampling phases, enough to
reach this threshold when every phase satisfies the density guarantee.

\begin{algorithm}[t]
\caption{Coloring by repeated density reduction}
\label{alg:coloring-halving}
\begin{algorithmic}[1]
\Require Matroids $M_1,\ldots,M_k$ on ground set $U$, with $k\geq2$
\State $W\gets U$, $\mathcal C\gets\varnothing$
\State Compute $\chimax\gets q(U)$
\State $\ell\gets\lceil\log_2 k\rceil$.
\For{$i=0,\ldots,\ell-1$}
    \State Compute $q\gets q(W)$. 
    \State If $q \leq \lceil\chimax/k\rceil$ then \textbf{break}.
    \State $L\gets\lceil8(k+1)q\rceil$, $T\gets\varnothing$
    \For{$p=1,\ldots,L$}
        \State Run AW with fresh randomness on $(M_j|_W)_{j\in[k]}$ at
        $\mathbf x=\mathbf1_W/q$, obtaining $I_p$
        \State Add $I_p\setminus T$ to $\mathcal C$ if it is nonempty
        \State $T\gets T\cup I_p$
    \EndFor
    \State $W\gets W\setminus T$
\EndFor
\State Color $W$ using the $k(k-1)$-approximation algorithm of~\cite{ArndtMPS26} and append its color classes to $\mathcal C$
\State \Return $\mathcal C$
\end{algorithmic}
\end{algorithm}

The algorithm always returns a feasible coloring. The phase limit
$\ell$ also bounds its running time regardless of the random outcomes.
The analysis shows that, with high probability, these phases reduce the
residual chromatic numbers to at most $h$, making the final coloring
inexpensive.

\subsection{One batch halves the density}

The following lemma is the main step. A single sample need not cover a
constant fraction of a small flat with high probability. Instead, we
track the number of newly covered elements across the whole batch.
As long as too many elements remain, the CRS moment bound continues to
force progress.

\begin{lemma}[Density reduction]
\label{lem:coloring-density-halving}
Fix $W\subseteq U$ and $q\geq\max\{1,\rho(W)\}$.
Let $I_1,\ldots,I_L$ be independent AW outputs on $(M_j|_W)_{j\in[k]}$
with fractional point $\mathbf1_W/q$, where
$L=\lceil8(k+1)q\rceil$, and let
$W'=W\setminus\bigcup_{p=1}^L I_p$.
If $q\geq20\log(2nk)$, then
\[
    \Pr[\rho(W')>q/2]\leq(2nk)^{-12}.
\]
\end{lemma}

\begin{proof}
Write $K=k+1$. As discussed, to bound the density of all sets, it suffices to consider flats. We first bound the probability that a fixed rank-$m$ flat
has more than $qm/2$ uncovered elements after the batch,
and then take a union bound over all flats.

\paragraph{Progress in one sample.}
Fix a rank-$m$ flat $F$ of $M_j|_W$, where $m\geq1$.
It starts with at most $qm$ elements. Our target is to leave at most
$qm/2$ uncovered.
Let $Y_p$ be the number left uncovered after $p$ samples, and let
$X_p$ be the number newly covered by sample $p$:
\[
    Y_p=\left|F\setminus\bigcup_{a=1}^p I_a\right|,
    \qquad X_p=Y_{p-1}-Y_p.
\]
In particular, $Y_0=|F|\leq qm$ and $Y_L=|F\cap W'|$.
Let $\mathcal H_p$ be the history of the first $p$ samples.
Given $\mathcal H_{p-1}$, the uncovered part of $F$ is fixed and the
next CRS run uses fresh randomness. Apply~\eqref{eq:cr-aw-transform} with coefficients indicating that uncovered set. That is, 
$a_e=1$ if $e\in F\setminus\bigcup_{i<p}I_i$ and $a_e=0$ otherwise. Also choose $\theta=\log 2$.
The fractional mass is $Y_{p-1}/q$, so
\begin{equation}
    \mathbb E[2^{-X_p}\mid\mathcal H_{p-1}]
    \leq(3/2)^{-Y_{p-1}/(Kq)}.
    \label{eq:coloring-batch-moment}
\end{equation}
The factor $2^{-X_p}$ records progress: each newly covered element
contributes another factor of $1/2$.

\paragraph{Progress over the batch.}
To track progress across the entire batch, define a potential that is exponential in the number of uncovered elements, and set it to zero once the target is reached. 
\[
    \Psi_p=2^{Y_p} \cdot \mathbf1_{\{Y_p>qm/2\}}.
\]
If $Y_{p-1}\leq qm/2$, then $Y_p\leq qm/2$ as well,
so the potential remains zero.  Otherwise, $Y_{p-1}>qm/2$ and each newly covered element reduces the potential by a
factor of $2$, unless the target is reached, in which
case the potential becomes zero. Equation~\eqref{eq:coloring-batch-moment} therefore gives
\[
    \mathbb E[\Psi_p\mid\mathcal H_{p-1}]
    \leq(3/2)^{-m/(2K)}\Psi_{p-1}.
\]
Repeating it over all $L$ samples, and using
$\Psi_0\leq2^{qm}$, yields
\[
    \mathbb E[\Psi_L]\leq2^{qm}(3/2)^{-Lm/(2K)}.
\]
If more than $qm/2$ elements of $F$ remain uncovered
after the batch, then $\Psi_L>2^{qm/2}$. Consequently,
\begin{align}
    \Pr[|F\cap W'|>qm/2]
    &\leq2^{-qm/2}\mathbb E[\Psi_L]\notag\\
    &\leq2^{qm/2}(3/2)^{-Lm/(2K)}
      \leq e^{-qm},
    \label{eq:coloring-flat-tail}
\end{align}
where the last step uses $L\geq8Kq$ and
$\tfrac12\log 2-4\log(3/2)<-1$.

\paragraph{All sets simultaneously.}
Note that if the residual density over all sets exceeds $q/2$, then there
must in particular be a flat $F$ of rank $m$ with
$|F\cap W'|>qm/2$.  Indeed, if $A\subseteq W'$ satisfies
$|A|>(q/2)r_j(A)$, its closure $F$ in $M_j|_W$ has the same rank
$m=r_j(A)$ and satisfies $|F\cap W'|>qm/2$.
Every rank-$m$ flat is the closure of an independent set of size $m$,
so there are at most $\binom{n}{m}\leq n^m$ such flats.
Union bounding~\eqref{eq:coloring-flat-tail} over ranks and matroids,
\[
    \Pr[\rho(W')>q/2]
    \leq k\sum_{m=1}^n n^m e^{-qm}
    \leq k\sum_{m=1}^{\infty}(2nk)^{-19m}
    \leq(2nk)^{-12},
\]
where the second inequality uses $q\geq20\log(2nk)$.
\end{proof}

The sum begins at rank one. Thus the same batch reduces the density of
all sets; no separate treatment of small-rank sets is needed.

\subsection{The approximation guarantee}

\begin{proof}[Proof of Theorem~\ref{thm:coloring-crs}]
We first count the colors assuming every phase satisfies the density
guarantee, and then bound the probability of a failure.

\paragraph{Number of colors.}
Let $s\leq\ell$ be the number of executed phases, and let $q_i$ be the
value of $q(W)$ at the start of phase $i$, indexing phases from zero.
Each successful phase leaves density at most $q_i/2$,
so the next individual chromatic numbers are at most
$\lceil q_i/2\rceil$. Inductively,
\[
    q_i\leq\left\lceil\frac{\chimax}{2^i}\right\rceil.
\]
Thus after at most $\ell$ phases, or upon early termination, the residual
set $W$ satisfies
\[
    q(W)\leq\left\lceil\frac{\chimax}{k}\right\rceil=h.
\]
The sampled color classes cost at most
\[
    \sum_{i=0}^{s-1}\left\lceil8(k+1)q_i\right\rceil
    \leq8(k+1)(2\chimax+\ell)+\ell
    =O(k\chimax).
\]
Here $s\leq\ell\leq k\leq\chimax$ under the theorem's assumption.
The final coloring uses at most
\[
    k(k-1)q(W)
    \leq k(k-1)\left\lceil\frac{\chimax}{k}\right\rceil
    \leq(k-1)\chimax+k(k-1)
    =O(k\chimax)
\]
additional colors. Since $\chimax\leq\OPT$, the total is
$O(k)\OPT$.

\paragraph{Success probability.}
Take $C\geq20$ in the theorem's assumption.
Every executed phase has
\[
    q=q(W)>h\geq\chimax/k\geq20\log(2nk).
\]
Conditional on the preceding history, the inputs to that phase are fixed
and its samples use fresh randomness. Lemma~\ref{lem:coloring-density-halving}
therefore applies with conditional failure probability at most
$(2nk)^{-12}$. There are at most $\ell=\lceil\log_2 k\rceil$ phases,
so all phases satisfy their density guarantee with probability at least
\[
    1-\ell(2nk)^{-12}\geq1-n^{-10}.
\]

\paragraph{Running time.}
$q(W)$ is computable in polynomial time, every AW call and the
final coloring algorithm run in polynomial time, and there are at most
$O(\log k)$ phases with $O(kn)$ calls each, since $q(W)\leq\chimax\leq n$.
This proves the running-time bound and completes the proof.
\end{proof}

\section{Application to Monotone Submodular Maximization}\label{sec:submodular}
In this section, we apply \Cref{thm:main-intro} to obtain new guarantees for constrained monotone submodular maximization (see \Cref{thm:submax_main_temp}). In particular we give a polynomial-time bicriteria approximation algorithm for maximizing a monotone submodular function $f$ subject to $k$ matroid
constraints, and $p$ packing and $c$ covering
constraints, extending~\cite{MizrachiSSU19} which handles a single matroid and constantly many packing and covering constraints. We require the constraints to be \emph{loose}: each packing constraint has a right-hand value at least $\Omega(\log p)$ and each covering constraint has right-hand value at least $\Omega(k\log(c))$. Under these conditions, we show that the continuous
greedy algorithm~\cite{CalinescuChekuriPalVondrak2011} followed by the AW CRS, outputs a set $S$ satisfying all packing
constraints without any violation, and all covering constraints up to an $O(k)$
factor, simultaneously with high probability in $p$ and $c$. Furthermore, the value of $f(S)$ is a $\left(\frac{1-1/e-\varepsilon}{k+1}\right)$-approximation to the optimal objective in expectation. This is formally stated in \Cref{thm:submax_main}.

\paragraph{Technical Overview}

We bound the probability that any packing constraint is violated or any covering constraint falls below its relaxed target. The packing and covering constraints are structurally different in this respect. For packing,
since the AW CRS first draws a random set $R$ by independently including
each element $e$ with probability $\bar{x}_e$ (where $\bar{x}$ is the
fractional point returned by continuous greedy), and $S \subseteq R$, a
standard Bernstein inequality applied to $R$ shows that each packing constraint $\sum_{e \in S} A_{ie} \leq b_i$ is satisfied with failure
probability at most $1/(p+2)^5$, allowing a union bound over the $p$ packing constraints to give feasibility with high probability.

For covering, each constraint $\sum_{e \in S}
C_{ie} \geq d_i$ is a lower bound on a linear function $\mathbf{C}_i
\cdot \mathbf{1}_S$ of the output, where $\mathbf{C}_i \in [0,1]^U$ is the
$i$-th row of covering constraint matrix $C$. This is where our main concentration theorem
is needed: for $\eps\in(0,1/5)$, \Cref{thm:main-intro} bounds the probability that row $i$ receives less than
$(1-\eps)(1-\eps/2)d_i/(k+1)$ by
$\exp\!\bigl(-\Omega(\eps^2 d_i/(k+1))\bigr)$.
With $d_i\geq 24(k+1)\eps^{-2}\log(c+2)$, this probability is at most
$1/(c+2)^5$, allowing a union bound over the $c$ constraints as before.

The $\Omega(k)$ violation of covering constraints is tight. One can encode a
$k$-dimensional matching instance via $k$ partition matroids and a single
covering constraint on matching size; any violation below $o(k)$ would yield
an $\frac{1}{o(k)}$-approximation for $k$-dimensional matching which
does not exist unless $\text{NP} \subseteq \text{BPP}$~\cite{LeeST2025}. 

\medskip
\textbf{Organization of the Section.} In \Cref{subsec:submax_prelim}, we discuss preliminaries for our result. In \Cref{subsec:submax_thm_alg}, we state our main theorem and algorithm. In \Cref{subsec:submax_analysis}, we analyze the algorithm and proved the stated guarantees. In \Cref{subsec:submax_context}, we provide context for our result, by showing the problem becomes NP-hard upon removal of any main assumption.

\subsection{Preliminaries}\label{subsec:submax_prelim}
Our algorithm combines two existing algorithmic tools: the continuous greedy method
for optimizing the multilinear extension of a submodular function, and the AW CRS for rounding fractional points in the intersection of $k$ matroid polytopes. We recall the relevant guarantees below.

\begin{definition}[Submodular Function]
A set function $f : 2^U \rightarrow \mathbb{R}_{\geq 0}$ is \textbf{submodular} if $f(A) + f(B) \geq f(A \cup B) + f(A \cap B)$ for all $A, B \subseteq U$. $f$ is \textbf{monotone} if $f(A) \leq f(B)$ for all {$A\subseteq B\subseteq U$} and \textbf{normalized} if $f(\emptyset) = 0$.
\end{definition}

\begin{definition}[Multilinear Extension \cite{CalinescuChekuriPalVondrak2011}]
The \textbf{multilinear extension} $F : [0, 1]^U \rightarrow \mathbb{R}_{\geq 0}$ of submodular function $f : 2^U \rightarrow \mathbb{R}_{\geq 0}$ is

\[ F(\mathbf{x}) = \sum_{S \subseteq U} f(S)\prod_{e \in S} x_e \prod_{e \notin S} (1-x_e) \]
\end{definition}

The multilinear extension $F$ enables continuous relaxation of submodular maximization. The following theorem guarantees that it can be approximately maximized over any solvable polytope in polynomial time.

\begin{theorem}\cite{CalinescuChekuriPalVondrak2011}\label{thm:cont_greedy}
For all fixed $\eps > 0$, given a monotone submodular set function $f$ and a general solvable polytope $P\subseteq[0,1]^U$, there is a polynomial-time algorithm to compute a solution $\mathbf{x}$ to the problem

\[ \max\{F(x): x \in P\}\]

such that $F(x) \geq (1-1/e-\eps)F(x^*)$, where $x^*$ is the optimal solution to the above problem.
\end{theorem}

Once a near-optimal fractional point is found, the AW CRS rounds it to a common independent set while preserving a $\tfrac{1}{k+1}$
fraction of the objective in expectation. For the monotone case, adding an element never decreases the objective, so the output equals the rounded set directly.

\begin{theorem}\cite{AdamczykWlodarczyk2018}
Let $f : 2^U \rightarrow \mathbb{R}_{\geq 0}$ be a non-negative submodular function with $f(\emptyset) = 0$. Let $M_i = (U, \mathcal{I}_i)$ for $i = 1, 2, \dots, k$ be $k$ matroids on common ground set $U$. Let $x \in P(M_i)$ for all $i \in [k]$ be arbitrary. Initialize $X = \emptyset$. Run the AW algorithm on $x$ and $M_i, i \in [k]$, and for each element $e$ added to the output set $S$, add element $e$ to $X$ if and only if $f(X \cup \{e\}) \geq f(X)$.\footnote{In \cite{AdamczykWlodarczyk2018} this condition is $f(X \cup \{e\}) > f(X)$, but it is easy to see that this can be relaxed to a non-strict inequality.} Then

\[ \mathbb{E}[f(X)] \geq \frac{1}{k+1}F(x) \]
\end{theorem}

As an immediate corollary, if $f$ is monotone submodular, then the condition $f(X \cup \{e\}) \geq f(X)$ is always satisfied and we obtain $X = S$.

\begin{corollary}\cite{AdamczykWlodarczyk2018}\label{cor:AW_multi}
Let $f : 2^U \rightarrow \mathbb{R}_{\geq 0}$ be a \textbf{monotone} non-negative submodular function with $f(\emptyset) = 0$. Let $M_i = (U, \mathcal{I}_i)$ for $i = 1, 2, \dots, k$ be $k$ matroids on common ground set $U$. Let $x \in P(M_i)$ for all $i \in [k]$ be arbitrary. Let $S$ be the output of the AW algorithm on $x$ and $M_i, i \in [k]$. Then

\[ \mathbb{E}[f(S)] \geq \frac{1}{k+1}F(x) \]
\end{corollary}

{
\subsubsection{Standard probability inequalities.}
The packing analysis uses the following two standard inequalities.

\begin{theorem}[Bhatia--Davis inequality~\cite{BhatiaDavis2000}]
\label{thm:bhatia_davis}
Let $X$ be a random variable taking values in $[a,b]$, and let
$\mu=\mathbb E[X]$. Then
\[
\operatorname{Var}(X)\leq(b-\mu)(\mu-a).
\]
\end{theorem}

\begin{theorem}[Bernstein's inequality~\cite{BoucheronLugosiMassart2013}]
\label{thm:bernstein}
Let $X_1,\ldots,X_n$ be independent random variables taking values in
$[0,1]$, and put $\nu^2=\sum_{i=1}^n\operatorname{Var}(X_i)$.
For every $t>0$,
\[
\Pr\!\left[\sum_{i=1}^n X_i-\mathbb E\!\left[\sum_{i=1}^n X_i\right]
\geq t\right]
\leq\exp\!\left(-\frac{t^2}{2(\nu^2+t/3)}\right).
\]
\end{theorem}
\par}

\subsection{Main Theorem and Algorithm}\label{subsec:submax_thm_alg}

\begin{theorem}\label{thm:submax_main}
Let $\eps \in (0, 1/5)$ be arbitrary. Let $f : 2^U \rightarrow \mathbb{R}_{\geq 0}$ be a monotone non-negative submodular function with $f(\emptyset) = 0$. Let $M_i = (U, \mathcal{I}_i)$ for $i = 1, 2, \dots, k$ be $k$ matroids on common ground set $U$. Let $A \in [0, 1]^{p \times U}$ be a nonnegative matrix and $b \in \mathbb{R}^p$ be a nonnegative vector s.t. $b_i \geq 40\eps^{-2}\log (p+2)$ for all $i \in [p]$. Let $C \in [0, 1]^{c \times U}$ be a non-negative matrix and $d \in \mathbb{R}^c$ be a non-negative vector s.t. $d_i \geq 24(k+1)\eps^{-2}\log(c+2)$ for all $i \in [c]$. If 

\[ P = \{x \in \mathbb{R}^U : x \in P(M_i) \forall i \in [k], Ax \leq b, Cx \geq d\} \]

is feasible, then there exists a randomized polynomial-time $\left(\frac{1-1/e-\eps}{k+1}\right)$-approximation algorithm in expectation for the problem

\[ \max\{f(X) : X \in \mathcal{I}_i \forall i \in [k], A\mathbf{1}_X \leq b, C\mathbf{1}_X \geq d\} \]

which always outputs a common independent set $S$.
With probability at least $1-1/(p+2)^4-1/(c+2)^4$, this set satisfies
$A\mathbf{1}_S\leq b$ and $C\mathbf{1}_S\geq \frac{c_\eps}{k+1}d$,
where $c_\eps=(1-\eps)(1-\eps/2)>0$ and $\mathbf{1}_S$ is the indicator vector of $S$.
\end{theorem}

\textbf{Algorithm:} Compute a $(1-1/e-\eps/2)$-approximate solution $\bar{x}$ to $\max\{F(x) : x \in (1-\eps/2)P\}$ where $F(x)$ is the multilinear extension of $f(X)$ (\Cref{thm:cont_greedy}). Let $S$ be the output of the AW algorithm on matroids $M_1, M_2, \dots, M_k$ and fractional point $\bar{x}$. Output $S$.

\subsection{Analysis}\label{subsec:submax_analysis}

We verify the three components of \Cref{thm:submax_main} separately.
\Cref{lemma:submax_approx} establishes the approximation ratio in expectation.
\Cref{lemma:submax_packing} shows the packing constraints are satisfied
exactly with high probability.  \Cref{lemma:submax_covering} shows the
covering constraints are satisfied approximately with high probability, and is the step where \Cref{thm:main-intro} is invoked. The proof of \Cref{thm:submax_main} then follows by a union bound over the latter two.

\begin{lemma}\label{lemma:submax_approx}
The algorithm is a $\left(\frac{1-1/e-\eps}{k+1}\right)$-approximation in expectation to the optimal objective value for the given problem.
\end{lemma}

\begin{proof}
Let $\bar{x}^*$ be the optimal solution to $\max\{F(x) : x \in (1-\eps/2)P\}$, $x^*$ be the optimal solution to $\max\{F(x) : x \in P\}$, and $\text{OPT} \subseteq U$ be the optimal solution to the original problem. Then

\begin{align}
F(\bar{x}) &\geq (1-1/e-\eps/2)F(\bar{x}^*) \label{eq:multi_1} \\
&\geq (1-1/e-\eps/2)F((1-\eps/2)x^*) \label{eq:multi_2} \\
&\geq (1-1/e-\eps/2)(1-\eps/2)F(x^*) \label{eq:multi_3} \\
&\geq (1-1/e-\eps)f(\text{OPT}) \label{eq:multi_4}
\end{align}

Line (\ref{eq:multi_1}) follows by \Cref{thm:cont_greedy}. Line (\ref{eq:multi_2}) follows by optimality of $\bar{x}^*$ for $\max\{F(x) : x \in (1-\eps/2)P\}$ and the feasibility of $(1-\eps/2)x^*$ for this problem. Line (\ref{eq:multi_3}) follows by the well-known fact that the multilinear extension of a monotone submodular function has the property $F(\lambda x) \geq \lambda F(x)$ for all $x \in [0,1]^U$ and $\lambda \in [0, 1]$. Line (\ref{eq:multi_4}) follows by $F(x^*) \geq f(\text{OPT})$ and simplification.

Thus by \Cref{cor:AW_multi}, we obtain

\[ \mathbb{E}[f(S)] \geq \frac{1}{k+1}F(\bar{x}) \geq\frac{1-1/e-\eps}{k+1} \cdot f(\text{OPT}) \]
\end{proof}

\begin{lemma}\label{lemma:submax_packing}
The packing constraints $A\mathbf{1}_S \leq b$ are satisfied with probability at least $1-1/(p+2)^4$.
\end{lemma}
\begin{proof}
We will show that $A\mathbf{1}_R \leq b$ with probability at least $1-1/(p+2)^4$ where 
$R = R(\bar{x})$ is the random set produced in the AW algorithm. This is sufficient 
because $S \subseteq R$. Consider an individual constraint $\sum_{j\in R} A_{ij}\leq b_i$ for some $i \in [p]$. Point $\bar{x}$ satisfies 
$\sum_{j=1}^n A_{ij} \bar{x}_j \leq (1-\eps/2)b_i$. Let $X_j = 1$ 
if $j \in R$, $0$ otherwise. Let $Y_j = A_{ij}$ if $j \in R$, $0$ otherwise. Then 
$\sum_{j \in R} A_{ij} = \sum_{j=1}^n Y_j$ is a sum of independent random variables in $[0, 1]$. 
Via Bernstein's inequality (\Cref{thm:bernstein}),
\[ \Pr\left[\sum_{j=1}^n Y_j - \mathbb{E}\left[\sum_{j=1}^n Y_j\right] \geq t\right] 
\leq \exp\left(-\frac{t^2}{2(\text{Var}\left(\sum Y_j\right) + t/3)}\right)\]
We have
\[ \mathbb{E}\left[\sum_{j=1}^n Y_j\right] = \sum_{j=1}^n \mathbb{E}\left[Y_j\right] 
= \sum_{j=1}^n A_{ij} \bar{x}_j \leq (1-\eps/2)b_i \]
and
\[ \text{Var}\left(\sum_{j=1}^n Y_j\right) = \sum_{j=1}^n \text{Var}\left(Y_j\right) 
\leq \sum_{j=1}^n \left(A_{ij} - A_{ij}\bar{x}_j\right)
\left(A_{ij}\bar{x}_j\right) \leq \sum_{j=1}^n A_{ij}
\bar{x}_j \leq (1-\eps/2)b_i \]
via Bhatia-Davis inequality (\Cref{thm:bhatia_davis}). Thus by Bernstein's inequality 
(\Cref{thm:bernstein}) we have
\begin{align*}
\Pr\left[\sum_{j \in R} A_{ij} > b_i\right] 
&= \Pr\left[\sum_{j=1}^n Y_j - (1-\eps/2)b_i > 
   (\eps/2) b_i\right] \\
&\leq \Pr\left[\sum_{j=1}^n Y_j - \mathbb{E}\left[\sum_{j=1}^n Y_j\right] > 
   (\eps/2) b_i\right] \\
&\leq \exp\left(-\frac{(\eps/2)^2 b_i^2}
   {2\left(\text{Var}\left(\sum Y_j\right) + (\eps/2) b_i/3\right)}
   \right) \\
&\leq \exp\left(-\frac{(\eps/2)^2 b_i^2}{2 b_i}
   \right) \\
&= \exp\left(-\frac{\eps^2 b_i}{8}\right) \\
&\leq \exp\left(-\frac{40 \log (p+2)}{8}\right) \\
&= \exp\left(-5\log (p+2)\right) \\
&= \frac{1}{(p+2)^5}
\end{align*}
There are $p$ packing constraints, so a union bound gives a success on all 
constraints with probability at least $1-1/(p+2)^4$.
\end{proof}

The packing argument above required nothing beyond $S \subseteq R(\bar{x})$ and
independence of the elements of $R(\bar{x})$; no matroid structure was needed.
The covering argument, by contrast, invokes the concentration result of
\Cref{thm:main-intro} directly, as discussed above.

\begin{lemma}\label{lemma:submax_covering}
We have $C\mathbf{1}_S \geq \frac{c_\eps}{k+1} d$ for constant 
$c_{\eps} = (1-\eps)(1-\eps/2) > 0$
with probability at least $1-1/(c+2)^4$.
\end{lemma}
\begin{proof}

Fix a covering constraint $i\in[c]$, and let
$Q_i=\mathbf{C}_i\cdot\bar{x}$ be its fractional value.
Since $\bar{x}\in(1-\eps/2)P$, we have $Q_i\geq(1-\eps/2)d_i$.
The two factors in $c_\eps$ account for scaling the fractional solution by
$1-\eps/2$ and retaining a $(1-\eps)/(k+1)$ fraction of its value under the CRS.
The coefficients of $\mathbf{C}_i$ lie in $[0,1]$, so
\Cref{thm:main-intro} with $\delta=\eps$ gives
\[
\Pr\!\left[\sum_{e\in S}C_{ie}\leq(1-\eps)\frac{Q_i}{k+1}\right]
\leq\exp\!\left(-\frac{\eps^2 Q_i}{4(k+1)}\right).
\]
Because $c_\eps d_i\leq(1-\eps)Q_i$, it follows that
\begin{align*}
\Pr\!\left[\sum_{e\in S}C_{ie}<\frac{c_\eps d_i}{k+1}\right]
&\leq\exp\!\left(-\frac{\eps^2(1-\eps/2)d_i}{4(k+1)}\right)\\
&\leq\frac{1}{(c+2)^5}.
\end{align*}
The last inequality uses $d_i\geq24(k+1)\eps^{-2}\log(c+2)$
and $6(1-\eps/2)>5$ for $\eps\in(0,1/5)$.
There are $c$ covering constraints, so a union bound gives success on all
constraints with probability at least $1-1/(c+2)^4$.

\end{proof}

\begin{proof}[Proof of \Cref{thm:submax_main}]
The proof follows by \Cref{lemma:submax_approx} and a union bound on \Cref{lemma:submax_packing} and \Cref{lemma:submax_covering}.
\end{proof}

\subsection{Context for Result}\label{subsec:submax_context}
We now show that each assumption in \Cref{thm:submax_main} is necessary in the sense that removing any single condition makes desired guarantees intractable to obtain.

\paragraph{Approximation Factor:} It is impossible to achieve an approximation factor of $\frac{1}{o(k)}$, even if $f(X) = |X|$, the matroids $M_i, i \in [k]$ are partition matroids, and there are no packing and covering constraints provided $NP \not \subseteq BPP$, as this captures the $k$-dimensional matching problem \cite{LeeST2025}.

\paragraph{Packing Constraints:} We must add some assumption on the packing constraints, as otherwise $f(X) = |X|$ and the packing constraints alone can encode the maximum independent set problem in graphs, which is NP-hard to approximate within $n^{1-\eps}$ for any $\eps > 0$ \cite{Zuckerman2007}.

\paragraph{Covering Constraints:} We must add some assumption on the covering constraints, because otherwise it is NP-hard to approximate the problem within a factor of $o(\log c)$ (where an $O(1)$ factor is guaranteed with $k=0$ matroids) via a reduction from Set Cover \cite{Feige1998}. For the reduction, $x_i$ represents whether we include set $S_i$ in the output for $i \in [n]$, we have a target objective value $\alpha \geq 0$, a single packing constraint $\sum x_i \leq \alpha$, and $c$ covering constraints encoding the set cover constraints. If we can determine feasibility of the resulting problem over $x \in \{0, 1\}^n$, we can determine whether the set cover instance can achieve target objective $\alpha$. Note that the added assumption $\alpha \geq 40\eps^{-2}\ln 3 = \Omega(1)$ and the ability to violate covering constraints by an $O(1)$ factor does not affect the reduction.

\paragraph{Violation of Covering Constraints:} Unless $\text{NP} \subseteq \text{BPP}$, it not possible to compute a feasible solution to

    \[ \mathcal{S}(v(k)) = \left\{X \subseteq U: X \in \mathcal{I}_i \quad \forall i \in [k],~ A\mathbf{1}_X \leq b,~ C\mathbf{1}_X \geq \frac{d}{v(k)} \right\} \]

    for a violation function $v(k) = o(k)$ given feasibility of $P$. We will show that if one can always compute a feasible solution to $\mathcal{S}(v(k))$, then it yields a $\frac{1}{v(k)}$-approximation for the $k$-dimensional matching problem. This claim is sufficient because it is not possible to approximate $k$-dimensional matching within a factor of $\frac{1}{o(k)}$ unless $\text{NP} \subseteq \text{BPP}$~\cite{LeeST2025}.
        
    Indeed, suppose we are given a $k$-uniform $k$-partite hypergraph $H = (V, F)$. Let OPT be the size of a maximum matching in $H$. The matchings in $H$ can be exactly represented as the independent sets in the intersection of $k$ partition matroids $M_i, i \in [k]$ on common ground set $F$. Compute the maximum integer value of $\alpha \geq 0$ such that the polytope

    \[ P = \left\{x \in \mathbb{R}^F : x \in P(M_i) \forall i \in [k], \sum_{i \in F} x_i \geq \alpha \right\} \]
    
    is feasible via binary search. Observe $\alpha \geq \text{OPT}$. Next, compute a feasible solution to

    \[ \mathcal{S}(v(k)) = \left\{X \subseteq F : X \in \mathcal{I}_i \forall i \in [k], 
    |X| \geq \frac{\alpha}{v(k)} \right\} \]

    A feasible solution $X$ corresponds to a matching in $H$ such that $|X| \geq \frac{\alpha}{v(k)} \geq \frac{\text{OPT}}{v(k)}$. Thus $X$ is a $\frac{1}{v(k)}$-approximation to the $k$-dimensional matching problem on $H$. Thus, feasibility testing of $\mathcal{S}(v(k))$ can be used to compute a $\frac{1}{v(k)}$-approximation of the maximum matching in $H$.

\section{The Hypergraph \texorpdfstring{$(g, f)$}{(g, f)}-Factor Problem}\label{subsec:submax_hyper}

In this section, we formally state the hypergraph $(g, f)$-factor problem, a generalization of the Representative Subhypergraph problem from the introduction, and show how \Cref{thm:main-intro} can be used to solve it. Recall that CRS are actually not the best algorithmic tool for this problem, and one can achieve better results using more elementary methods (independent rounding of the natural linear program, analyzed using standard Chernoff and union bounds).

In the hypergraph $(g, f)$-factor problem, the input is a hypergraph $H = (V, F)$ and functions $g, f : V \rightarrow \mathbb{Z}_{\geq 0}$ s.t. $g(v) \leq f(v)$ for all $v \in V$. The output is a subhypergraph $H' \subseteq H$ s.t. $g(v) \leq \deg_{H'}(v) \leq f(v)$ for all $v \in V$ if one exists.

The feasible region of subhypergraphs $H'$ satisfying the degree conditions is given by

\[ \left\{x \in \{0, 1\}^F : g(v) \leq \sum_{e \ni v} x_e \leq f(v) \hspace{1em} \forall v \in V\right\} \]

\begin{theorem}
{There is an absolute constant $\alpha>0$ such that the
following holds.} Let $H = (V, F)$ be a $k$-uniform $k$-partite hypergraph
for $k \geq 2$ with functions $g, f : V \rightarrow \mathbb{Z}_{\geq 0}$
such that {$\alpha k\log(|V|+2) \leq g(v) \leq f(v)$ for
all $v\in V$}. Assume the hypergraph $(g, f)$-factor problem on $H, g, f$
is feasible. Then there exists a randomized polynomial-time algorithm
to compute a subhypergraph $H' \subseteq H$ such that
$\frac{g(v)}{20k} \leq \deg_{H'}(v) \leq f(v)$ for all $v \in V$
with probability at least {$1-1/(|V|+2)^4$}.
\end{theorem}

\begin{proof}\label{thm:(g,f)_main}
We will use partition matroids to encode the upper bound constraints and then apply \Cref{thm:submax_main}. Let $V_1, V_2, \dots, V_k$ be the parts of $V$. Define partition matroids $M_1, M_2, \dots, M_k$ on common ground set $F$, where $M_i$ has a part for each vertex $v \in V_i$ with capacity $f(v)$ containing the hyperedges incident to $v$ in $H$ for all $i \in [k]$. Then the feasible region of subhypergraphs $H'$ satisfying the degree conditions is given by

\[ \left\{
x \in \{0,1\}^F :
x \in P(M_i) \quad \forall i \in [k],\;
\sum_{e \ni v} x_e \ge g(v) \quad \forall v \in V
\right\} \]

{By assumption, this feasible region is nonempty. There
are $c=|V|$ covering constraints $Cx\geq d$, with
$C_{v,e}=\mathbf1_{\{v\in e\}}$ and $d_v=g(v)$.
Choose $\eps=1/10$ and $\alpha=3600$. Since $k+1\leq3k/2$ for $k\geq2$,
\[
d_v\geq3600k\log(c+2)
\geq24(k+1)\eps^{-2}\log(c+2).
\]
Apply the algorithm of \Cref{thm:submax_main} with these covering
constraints, no packing constraints, and the identically zero objective.
Its output $S$ is independent in every $M_i$, so all upper degree bounds
hold deterministically. By \Cref{lemma:submax_covering}, with probability
at least $1-1/(c+2)^4$,
\[
C\mathbf1_S\geq\frac{c_\eps}{k+1}d
=\frac{171}{200(k+1)}d\geq\frac{d}{20k},
\]
where $c_\eps=(1-\eps)(1-\eps/2)=171/200$ and the last inequality uses
$k+1\leq3k/2$. Taking $H'=(V,S)$ completes the proof.}
\end{proof}

\section{Pessimistic estimators and gain--loss martingales}
\label{app:pessimistic-estimators}

This appendix makes explicit the connection between our concentration
argument and the pessimistic-estimator method of Harvey and
Olver~\cite{HarveyOlver2014}. The common structure is a process of
feasible fractional or randomized states, together with a nonnegative
quantity that certifies failure at termination and whose conditional
expectation does not increase. The mechanism that gives this drift is
different: pipage rounding uses concavity along swaps, whereas our
analysis of AW uses the one-step balance between gain and residual loss.

\subsection{The pessimistic estimator for pipage rounding}

Let $y\in[0,1]^E$ and let $R_y$ be the product-random subset that
includes each element $e$ independently with probability $y_e$. For
$a\in[0,1]^E$, a lower threshold $v$, and $\theta>0$, define
\begin{equation}\label{eq:appendix-pe-chernoff}
q_{v,\theta}(y)
:=\E\!\left[e^{\theta(v-a(R_y))}\right]
=e^{\theta v}\prod_{e\in E}
  \bigl(1-y_e(1-e^{-\theta a_e})\bigr).
\end{equation}
Because $\mathbf1_{\{a(R_y)\le v\}}\le
e^{\theta(v-a(R_y))}$, this is a pessimistic estimator:
\[
\Pr(a(R_y)\le v)\le q_{v,\theta}(y).
\]
At an integral point $y=\mathbf1_X$ it takes the value
$q_{v,\theta}(\mathbf1_X)=e^{\theta(v-a(X))}$, which is at least one
whenever $a(X)\le v$. Thus it both bounds the independent-rounding
failure probability at the initial fractional point and certifies a
bad final integral solution.

Now consider one pipage step from $y$ to feasible endpoints $y^-$ and
$y^+$ on a swap line, with
$y=\lambda y^-+(1-\lambda)y^+$ for $\lambda\in[0,1]$. Along a swap
direction $e_i-e_j$, every factor in~\eqref{eq:appendix-pe-chernoff}
is constant except the $i$th and $j$th; their product is a quadratic
with nonpositive second derivative. Hence $q_{v,\theta}$ is concave
along the line, and
\begin{equation}\label{eq:appendix-pipage-endpoints}
\min\{q_{v,\theta}(y^-),q_{v,\theta}(y^+)\}
\le \lambda q_{v,\theta}(y^-)+(1-\lambda)q_{v,\theta}(y^+)
\le q_{v,\theta}(y).
\end{equation}
The two consequences correspond to the two forms of pipage rounding.
Deterministic pipage can choose the endpoint with smaller estimator
value, so the value never increases. Randomized pipage chooses $y^-$
with probability $\lambda$ and $y^+$ with probability $1-\lambda$;
this preserves the current point in conditional expectation and gives
\[
\E[q_{v,\theta}(y^{t+1})\mid y^t]
=\lambda q_{v,\theta}(y^-)+(1-\lambda)q_{v,\theta}(y^+)
\le q_{v,\theta}(y^t).
\]
Consequently $q_{v,\theta}(y^t)$ is a nonnegative supermartingale
along the randomized sequence. If $y^T=\mathbf1_X$ at termination,
then~\eqref{eq:appendix-pe-chernoff} and the terminal certificate give
\[
\Pr(a(X)\le v)
\le \E[q_{v,\theta}(y^T)]
\le q_{v,\theta}(y^0).
\]
The right-hand side is precisely the usual product Chernoff estimator
evaluated at the starting marginals. This is the key transfer: the
product estimator is designed for independent rounding, but swap
concavity makes its expected value decrease under the dependent
rounding trajectory. Harvey and Olver develop this principle well
beyond scalar Chernoff bounds, including estimators for submodular and
matrix concentration.

\subsection{A state-based estimator for AW}

The AW trajectory has no swap line and no pair of feasible endpoints
from which to select a favorable state. The endpoint inequality
\eqref{eq:appendix-pipage-endpoints} is therefore not available. AW
instead carries two adapted nonnegative quantities: the value already
collected, $V_t$, and the residual weight, $Z_t$. Write
\[
G_t=V_{t+1}-V_t,\qquad L_t=Z_t-Z_{t+1}.
\]
For one-matroid AW and $a\in[0,1]^E$, Lemma~\ref{lem:one-matroid-gain-loss}
gives $0\le G_t,L_t\le1$ and
$\E[L_t\mid\mathcal F_t]\le2\E[G_t\mid\mathcal F_t]$, with
$(V_0,Z_0)=(0,Q)$ and $(V_n,Z_n)=(a(S),0)$, where $Q=a\cdot x$.
The residual loss is not controlled pathwise by the gain; its conditional
mean is controlled instead. That is the input that replaces swap
concavity.

For $\theta>0$, put
\[
c=1-e^{-\theta},\qquad
\eta=\log\left(1+\frac c2\right),\qquad
g_t=e^{\theta v}\exp(-\theta V_t-\eta Z_t).
\]
The exponential chord bounds on $[0,1]$ give
\[
e^{-\theta G_t}\le1-cG_t,\qquad
e^{\eta L_t}\le1+\frac c2L_t.
\]
Multiplying, using $(1-cG_t)(1+cL_t/2)
\le1-cG_t+cL_t/2$, and then taking conditional expectation yields
\begin{align*}
\E[g_{t+1}\mid\mathcal F_t]
&=g_t\E[e^{-\theta G_t+\eta L_t}\mid\mathcal F_t]\\
&\le g_t\left(1-c\E[G_t\mid\mathcal F_t]
             +\frac c2\E[L_t\mid\mathcal F_t]\right)\\
&\le g_t.
\end{align*}
Thus $g_t$ is a nonnegative supermartingale. At termination,
$g_n=e^{\theta(v-a(S))}\ge1$ whenever $a(S)\le v$, while
$g_0=e^{\theta v-\eta Q}$. The same terminal-event argument therefore
gives
\begin{equation}\label{eq:appendix-aw-estimator}
\Pr(a(S)\le v)\le\E[g_n]\le
\exp\left(\theta v-Q\log\left(1+\frac{1-e^{-\theta}}2\right)\right).
\end{equation}
Optimizing~\eqref{eq:appendix-aw-estimator} at
$v=(1-\delta)Q/2$ gives the one-matroid lower-tail bound in
Corollary~\ref{cor:one-matroid-lower-tail}. The general gain--loss
theorem~\ref{thm:gain-loss-concentration} carries out the same argument
for arbitrary bounds $(\alpha,\rho,B)$ and yields a tail bound with no
dependence on the number of arrival steps.

\subsection{What the comparison says}

Both methods follow the same three-step argument: find a nonnegative
process $g_t$ with $g_\tau\ge1$ on failure, prove
$\E[g_{t+1}\mid\mathcal F_t]\le g_t$, and conclude
$\Pr(\text{failure})\le\E[g_\tau]\le g_0$. They differ in where the
drift comes from. Harvey--Olver use the product-distribution estimator
$q_{v,\theta}(y)$; swap concavity makes its conditional average decrease
under randomized pipage rounding and lets deterministic pipage choose
a nonincreasing endpoint. AW has no such endpoint choice. Its state-based
potential $g_t$ instead decreases in conditional expectation because
bounded gain and expected residual-loss control balance each update.
This gives concentration on the scale of the initial residual mass,
independent of the number of updates.

The methods therefore complement one another: a pessimistic estimator
specifies what certifies failure, while the gain--loss interface offers
another way to prove that certificate has nonpositive drift when the
rounding process has bounded gain and residual loss but no useful swap
geometry. In AW this produces a state-based estimator. It does not
replace swap concavity or imply that the product estimator itself works
for AW; it extends the supermartingale strategy to processes whose
natural structure is gain--loss rather than endpoint choice.

\end{document}